\documentclass[12pt,draftclsnofoot,onecolumn]{IEEEtran}
\usepackage{amsmath}
\usepackage{amsfonts}
\usepackage{amssymb}
\usepackage[latin1]{inputenc}
\usepackage[T1]{fontenc}
\usepackage{multirow}

\ifCLASSINFOpdf
   \usepackage[pdftex]{graphicx}
\else
   \usepackage[dvips]{graphicx}
\fi

\usepackage{subfig}
\usepackage{array}
\usepackage{fancyhdr}
\begin{document}

\pagestyle{plain}

\thispagestyle{fancy}

\newtheorem{remark}{Remark}
\newtheorem{proposition}{\textit{Proposition}}
\newtheorem{definition}{\textit{Definition}}
\newtheorem{problem}{\textit{Problem}}


\title{Delay-Distortion-Power Trade Offs in Quasi-Stationary
Source Transmission over Block Fading Channels}


\author{\IEEEauthorblockN{Roghayeh Joda$^{\star}$, Farshad~Lahouti$^{+}$ and Elza Erkip$^{*}$}\\
\IEEEauthorblockA{$\star$\small{Iran Telecommunication Research Center, Tehran, Iran, r.joda@itrc.ac.ir\\
$^+$School of Electrical and Computer Engineering, University of Tehran, Tehran, Iran, lahouti@ut.ac.ir\\
$^*$Dept. of Electrical and Computer Engineering, Polytechnic School
of Engineering, New York University, Brooklyn, NY, elza@nyu.edu}}
\thanks{Part of this work was done while Roghayeh Joda was visiting the Deptartment of Electrical and Computer Engineering, New York University, Polytechnic School
of Engineering.} }

\maketitle \thispagestyle{empty}

 \begin{abstract}
This paper investigates delay-distortion-power trade offs in
transmission of quasi-stationary sources over block fading channels
by studying encoder and decoder buffering techniques to smooth out
the source and channel variations. Four source and channel coding
schemes that consider buffer and power constraints are presented to
minimize the reconstructed source distortion. The first one is a
high performance scheme, which benefits from optimized source and
channel rate adaptation. In the second scheme, the channel coding
rate is fixed and optimized along with transmission power with
respect to channel and source variations; hence this scheme enjoys
simplicity of implementation. The two last schemes have fixed
transmission power with optimized adaptive or fixed channel coding
rate. For all the proposed schemes, closed form solutions for mean
distortion, optimized rate and power are provided and in the high
SNR regime, the mean distortion exponent and the asymptotic mean
power gains are derived. The proposed schemes with buffering exploit
the diversity due to source and channel variations. Specifically,
when the buffer size is limited, fixed channel rate adaptive power
scheme outperforms an adaptive rate fixed power scheme. Furthermore,
analytical and numerical results demonstrate that with limited
buffer size, the system performance in terms of reconstructed signal
SNR saturates as transmission power is increased, suggesting that
appropriate buffer size selection is important to achieve a desired
reconstruction quality.

\end{abstract}
\IEEEpeerreviewmaketitle
\section{Introduction}\label{SI}
Multimedia signals such as video exhibit quasi-stationary
characteristics, causing the compression rate to vary over time.
Wireless channels, on the other hand, also vary over time, making
wireless video transmission challenging. In order to maintain a
desired signal quality, multimedia communications over wireless
channels involve buffering at the encoder and decoder to smooth out
the source and channel variations at the cost of delay. For
delay-constrained communications the buffer size is kept limited,
and the transmitter controls the rate and/or transmission power to
minimize the end-to-end distortion, while preventing buffer overflow
and underflow. The goal of this paper is to study
delay-distortion-power trade offs in transmission of
quasi-stationary sources over block fading channels from the
perspective of source and channel code design and the associated
performance scaling laws.

There is a rich literature on source and channel coding for wireless
channels. The end-to-end mean distortion for transmission of a
stationary source over a block fading channel is considered in,
e.g., \cite{R11}\nocite{R15}\nocite{R18}-\cite{R27}, where the
performance is studied in terms of the (mean) distortion exponent or
the decay rate of the end-to-end mean distortion with (channel)
signal to noise ratio (SNR) in the high SNR regime. Delay-limited
communication of a stationary source over a wireless block fading
channel from the channel outage perspective is studied in
\cite{{R1}} and \cite{{R5}}. The transmission of a stationary source
over a MIMO block fading channel with constant power transmission is
considered in \cite{R21}, where the distortion outage probability
and the outage distortion exponent are considered as performance
measures. Several schemes for transmission of a quasi-stationary
source over a block fading channel are proposed in \cite{R22} to
minimize the distortion outage probability. The results demonstrate
the benefit of power adaption for delay-limited transmission of
quasi-stationary sources over wireless block fading channels from a
distortion outage perspective.

Considering delay-limited transmission of a quasi-stationary source
over a wireless block fading channel and noting the buffer
constraints, in this paper we propose a framework for rate and power
adaptation that uses source and channel codes achieving the
rate-distortion and the capacity in a given source and channel
state. Throughout, we assume that the channel state information is
known at the transmitter and the receiver. As described in Section
II, the end-to-end mean distortion, the mean distortion exponent and
asymptotic mean power gains are the performance metrics of interest.
Under average transmission power and buffer size constraints, four
designs are presented. The first scheme provides adaptation of
source and channel coding rates and the transmission power such that
the end-to-end mean distortion is minimized. The second scheme is a
channel optimized power adaptation strategy to minimize the mean
distortion for a given optimized fixed channel rate. This, for
example, could be useful when we are interested in simple
transmission schemes with a single channel (coding) rate. The other
two designs are constant power delay-limited communication schemes
with channel optimized adaptive or fixed rates.

The performance of the proposed schemes are evaluated and compared
both analytically and numerically. The scaling laws involving mean
distortion exponent and asymptotic mean power gain are derived in
the large SNR regime with limited or asymptotic buffer sizes. The
results demonstrate that the proposed schemes utilize the diversity
provided by increasing the number of source blocks in a frame
(buffer) at varying levels. Moreover, the presented schemes capture
a larger (source) diversity gain when the non-stationary
characteristic of the source is intensified or equivalently the
variations in the source characteristics from one block to another
increases. Another interesting observation is that with limited
buffer size and increasing transmission power, the system
performance in terms of reconstructed signal SNR saturates. In other
words, depending on the level of source variations, delay constraint
and the desired performance, the buffer size needs to be carefully
designed to ensure the performance scales properly with transmission
power. The results show that the proposed source and channel
optimized rate and power adaptive scheme outperforms other proposed
schemes in terms of the end-to-end mean distortion. For the case
that the buffer constraint is relatively small in comparison to the
power limit, it is seen that an adaptive power single channel rate
scheme outperforms a rate adaptive scheme with constant transmission
power. This is in contrast to the observation made in \cite{R23},
which is from the Shannon capacity perspective.

Note that the delay-limited transmission in \cite{R11}-\cite{R22}
refers to the scenario where each frame interval is short as it
spans only a limited number of channel blocks and the transmitter
cannot average out over channel fluctuations. Although we assume a
quasi stationary source as in \cite{R22}, we consider the end-to-end
mean distortion and the buffer size limitation, and as such the resulting
delay we investigate here is of a distinct nature and is primarily
affected by the variability of the source statistics.

The paper is organized as follows. Following the description of
system model in Section \ref{SII}, Section \ref{SIV} presents the
proposed scheme based on adaptive rate and power source and channel
coding design to minimize mean distortion. Next, in Section
\ref{SIII}, we present the scheme with adaptive power and fixed
channel coding rate optimized for minimized mean distortion. Two
constant power schemes are proposed in Section \ref{SV}. Performance
comparisons and evaluations are presented in Sections \ref{SVIB} and
\ref{SVI}. Finally the paper is concluded in Section \ref{SVII}.

\section{System Model}\label{SII}
We consider the transmission of a quasi-stationary source over a
block fading channel. The source is assumed to be independent
identically distributed circularly symmetric complex Gaussian with
zero mean and variance $\sigma^2_s$,
$s\in\mathcal{S}=\{1,2,...,N_s\}$ in a given source block of $N$
samples \cite{R2}. The source state $s\in\mathcal{S}$ is a discrete
random variable with the probability mass function (pmf) $p_s$. For
optimized source coding rate $R_s$ bits per source sample in state
$s$, the resulting distortion is $D=\sigma^2_s\,2^{-R_s}$ \cite{R3}.
One frame of the source is defined as $K$ source blocks. We assume
that the source sum rate in a frame is constrained to
$\tilde{B}_{\rm{max}}$ bits per frame due to buffer (delay)
limitations \cite{R25}\cite{R24}. Therefore, we have
$N{\sum_{j=1}^{K}R_{s_j}}\leq \tilde{B}_{\rm{max}}$, where ${s_j}$
is the source state in the source block $j$. Alternatively, we
obtain\vspace{-6pt}
\begin{align}\label{E47}
\frac{\sum_{j=1}^{K}R_{s_j}}{K}\leq B_{\rm{max}},
\end{align}
where
$B_{\rm{max}}\overset{\Delta}{=}\frac{\tilde{B}_{\rm{max}}}{KN}$,
and the LHS of \eqref{E47} indicates the average source coding rate
in bits per sample over a frame.

We consider a block fading channel for transmitting the source to
the destination. Let $X$, $Y$ and $Z$, respectively indicate channel
input, output and additive noise, where $Z$ is an i.i.d circularly
symmetric complex Gaussian noise, $Z\sim\mathcal{CN}(0,1)$.
Therefore, we have $Y=h X+Z$, where $h$ is the multiplicative
fading. The channel gain $\alpha=|h|^2$ is assumed to be constant
across one channel block and independently varies from one channel
block to another according to the continuous probability density
function $f(\alpha)$. For a Rayleigh fading channel, $\alpha$ is an
exponentially distributed random variable, where we assume
$E[\alpha]=1$. The instantaneous capacity of the fading Gaussian
channel over one channel block (in bits per channel use) is given by
$C(\alpha,\gamma)=\log_2(1+\alpha\gamma)$, where $\gamma$ is the
transmission power. We consider the long term power constraint
$\text{E}[\gamma]\leq \bar{P}$, where the expectation is taken over
the fading distribution \cite{R1}.

Each source frame is channel encoded to a single channel codeword
spanning one channel block. The bandwidth expansion ratio of the
system is denoted by $b$ channel uses per source sample, where $b\in
{\mathcal{R}}^+\textbackslash(0,1)$. Thus, the source rates in a
frame are assigned such that the following constraint is
satisfied\vspace{-18pt}
\begin{align}\label{E45}
\frac{\sum_{j=1}^{K}R_{s_j}}{K}=bR,
\end{align}
where $R$ in bits per channel use is the channel coding rate.
Equations \eqref{E47} and \eqref{E45} result in $bR\leq
B_{\rm{max}}.$
Assuming perfect knowledge at the transmitter and receiver of the
source and channel states in a frame, $\gamma$, $R_{s_j}$ and $R$
can be generally chosen as a function of source and channel states,
i.e., $\gamma=\gamma(\Sigma,\alpha)$,
$R_{s_j}=R_{s_j}(\Sigma,\alpha)$ and $R=R(\Sigma,\alpha)$, where
$\Sigma=[s_1,...,s_K]$ is the vector of source states in a given
frame. Thus, the end-to-end mean distortion in general is given
by\vspace{-12pt}
\begin{align}\label{E1}
\text{E}[D]=\text{E}_{\Sigma,\alpha}\left[\frac{\sum_{j=1}^{K}\sigma^2_{s_j}2^{-R_{s_j}}}{K}
I(R\leq C(\alpha,\gamma))+\frac{\sum_{j=1}^{K}\sigma^2_{s_j}}{K}
I(R>C(\alpha,\gamma))\right],
\end{align}
where $I$ is an indicator function which is $1$ when its argument is true and zero otherwise. Without loss of generality, we assume that
$\sigma^2_{s_1},...,\sigma^2_{s_K}$ are in a descending order. Note
that the problem and the proposed designs can be easily extended to
the case in which one source frame corresponds to more than one fading channel blocks.

We define the buffer constrained mean distortion exponent,
$\Delta^{B_{\rm{max}}}_{\textit{MD}}$, and the buffer unconstrained
mean distortion exponent, $\Delta_{\textit{MD}}$ as our performance
measures \cite{R11}, where\vspace{-8pt}
\begin{equation}\label{E83}
\Delta^{B_{\rm{max}}}_{\textit{MD}}=\underset{\bar{P}\rightarrow\infty}{\lim}{-\frac{\ln{\text{E}[D]}}{\ln{\bar{P}}}},\;\;\;\;\;\Delta_{\textit{MD}}=\underset{\bar{P}\rightarrow\infty,{B_{\rm{max}}}\rightarrow\infty}{\lim}{-\frac{\ln{\text{E}[D]}}{\ln{\bar{P}}}}.
\end{equation}

Let $\bar{P}_1$ and $\bar{P}_2$ be the average powers transmitted to
asymptotically achieve a specific buffer constrained or
unconstrained mean distortions by two different schemes. We define
the corresponding asymptotic mean power gain as
follows\vspace{-10pt}
\begin{equation}\label{E84}
G^{B_{\rm{max}}}_{M\!D}(\text{or }
G_{M\!D})=10\log_{10}{\bar{P}_2}-10\log_{10}{\bar{P}_1}.
\end{equation}\vspace{-6pt}
In the sequel, we also use the following mathematical definitions
and approximations (see (5.1.53) of \cite{R13})\vspace{-6pt}
\begin{equation}\label{E137}
E_p(x):={\int^\infty_1}{\frac{e^{-x\alpha}}{\alpha^p}\mathrm{d}\alpha};\;\;\;\;\;\;E_p(0)=\frac{1}{p-1},\;p>1;\;\;\;\;\;\;\Gamma\left(t,x\right):=\int_x^\infty{\alpha^{t-1}e^{-\alpha}\mathrm{d}\alpha},\;\;x,t\geq
0
\end{equation}\vspace{-10pt}
\begin{equation}\label{E98}
[x]^+:=\operatorname{\rm{max}}\{x,0\},\;\;\;\;\;\;E_1(x)\cong-E_c-\ln(x),
\;\;\;\;\;\;\;\;\;e^x\cong1+x \quad \text{if} \;x\rightarrow 0,
\end{equation}
where $E_c=0.5772156649$ is the Euler constant. 
In the following sections, we explore designs to minimize the
end-to-end mean distortion in the presence of average power and
buffer constraints. The optimization variables are power $\gamma$
and rate $R$ in each channel block, and source rates $R_{s_j}$ in
source block $j\in\{1,...,K\}$ over a frame, which may in
general be a function of channel and source states. Depending on
whether rate and/or power are adjusted or remain constant, we
present four schemes in the sequel.

\section{Source and Channel optimized Rate and Power Adaptation}\label{SIV}
In this section, we consider power and rate adaptation with regard
to source and channel states for improved performance of
communication of a quasi-stationary source over a block fading
channel. The objective is to devise power and rate adaptation
strategies for each block of the channel and the source such that
the end-to-end mean distortion is minimized, when the average power
and buffer are respectively constrained to $\bar{P}$ and
$B_{\rm{max}}$. The source and channel coding rates may be
controlled with respect to the channel state to avoid channel
outages such that $R\leq C(\alpha,\gamma)$. Thus, we have the
following design problem.
\begin{problem}\label{PR2}
The problem of delay-limited Source and Channel Optimized Rate and
Power Adaptation (SCORPA) for communication of a quasi-stationary
source with limited buffer over a block fading channel is formulated
as follows\vspace{-8pt}
\begin{align}\label{E2}
&\underset{\gamma(\Sigma,\alpha),{R}_{s_j}(\Sigma,\alpha)}{\operatorname{min}}\;{\text{E}_{\Sigma,\alpha}\left[\sum_{j=1}^{K}\sigma^2_{s_j}2^{-R_{s_j}}\right]}\;\text{
subject to }
\end{align}{\vspace{-20pt}}
\begin{align}
&
\text{E}[\gamma]\leq\bar{P},\;\;\;\;\;\;{\sum_{j=1}^{K}R_{s_j}}\leq
K B_{\rm{max}}
\end{align}{\vspace{-20pt}}
\begin{align}
&{\sum_{j=1}^{K}R_{s_j}}\leq {K} b C(\alpha,\gamma)\label{NewCons}.
\end{align}
\end{problem}

\begin{proposition}\label{P1}
Solution to Problem \ref{PR2}, denoted by $R^*_{s_j}$ and
$\gamma^*$, for a given frame and an arbitrary block fading channel
is given by\vspace{-6pt}
\begin{align}\label{E25}
&R^*_{s_j}=\begin{cases}0&\text{if}\;\alpha\leq\frac{\lambda}{bK\sigma_{s_1}^2}\\
\!\left[\log_2\frac{\sigma^2_{s_j}}{\lambda_2}\right]^+\!&\text{if}\;d_{1,
m}\leq\alpha< d_{2,m},\;\alpha< c_m:\!\forall{m}\in\{1,...,K\}\\
\!\left[\log_2\frac{\sigma^2_{s_j}}{\mathbf{\tilde{\lambda}}_2}\right]^+\!&\text{if}\;\!d_{1,m}\leq\alpha<
d_{2,m},\;\alpha\geq c_m:\forall{m}\in\{1,...,K\}
\end{cases}
\end{align}\vspace{-12pt}
\begin{align}\label{E26}
&\gamma^*=\begin{cases}0&\text{if}\;\alpha\leq\frac{\lambda}{bK\sigma_{s_1}^2}\\
\frac{\frac{\alpha\lambda_2bK}{\lambda}-1}{\alpha}&\text{if}\;d_{1,m}\leq\alpha< d_{2,m},\;\alpha<c_m:\forall{m}\in\{1,...,K\}\\
\frac{2^{\frac{B_{\rm{max}}}{b}}-1}{\alpha}&\text{if}\;d_{1,m}\leq\alpha<
d_{2,m},\;\alpha\geq c_m:\forall{m}\in\{1,...,K\},
\end{cases}
\end{align}
where\vspace{-10pt}
\begin{align}\label{E52}
&c_m=2^{\frac{m+bK}{m}\frac{B_{\rm{max}}}{b}}\frac{\lambda}{bK\sqrt[m]{\sigma^2_{s_1}...\sigma^2_{s_m}}}
\end{align}
\begin{align}\label{E53}
&d_{1,m}=\frac{\lambda\sqrt[bK]{\sigma^2_{s_1}...\sigma^2_{s_m}}}{bK\sqrt[bK]{\sigma_{s_m}^{2(m+bK)}}}\;\;\;d_{2,m}=\begin{cases}\frac{\lambda\sqrt[bK]{\sigma^2_{s_1}...\sigma^2_{s_m}}}{bK\sqrt[bK]{\sigma_{s_{m+1}}^{2(m+bK)}}}&\text{if}\;m\neq K\\
\infty&\text{if}\;m=K
\end{cases}
\end{align}
\begin{align}\label{E27}
\lambda_2=\sqrt[m+bK]{\left(\sigma^2_{s_1}...\sigma^2_{s_m}\right)}\sqrt[m+bK]{(\frac{\lambda}{\alpha
bK})^{bK}}, \;\;\;\;\mathbf{\tilde{\lambda}}_2=
\sqrt[n]{\frac{\sigma_{s_1}^2...\sigma^2_{s_n}}{2^{KB_{\rm{max}}}}}
\end{align}
$n$ is an integer in $\{1,...,K\}$ such that $ \sigma_{s_{n+1}}^2<
\mathbf{\tilde{\lambda}}_2 \leq \sigma_{s_{n}}^2$. The parameter
$\lambda$ is selected such that\vspace{-4pt}
\begin{align}\label{E48}
\text{E}_{\Sigma}\!\left[\!\sum_{m=1}^{K}\!\int_{a_{1,m}}^{a_{2,m}}{\!\frac{2^{\frac{B_{\rm{max}}}{b}}-1}{\alpha}\!}f(\alpha)d\alpha\!\!+\!\!\!\int_{e_{1,m}}^{e_{2,m}}
\!\!\left\{\!\sqrt[m+bK]{\sigma^2_{s_1}...\sigma^2_{s_m}}\!\!\sqrt[m+bK]{\left(\!\frac{bK}{\lambda}\!\right)^m}\!\alpha^{\frac{-bK}{m+bK}}\!-\!\alpha^{-1}\!\right\}f(\alpha)d\alpha\!\right]\!\!=\!\!\bar{P},
\end{align}\vspace{-12pt}
where\vspace{-4pt}
\begin{align}\label{E100}
a_{1,m}\!=\!\max\{d_{1,m},c_m\},\;\;a_{2,m}\!=\!\max\{d_{2,m},c_{m}\},\;\;e_{1,m}\!=\!d_{1,m},\;\;e_{2,m}\!=\!\max\{\!\min\{c_m,d_{2,m}\},e_{1,m}\!\}.
\end{align}
Specifically for Rayleigh block fading channel, \eqref{E48} is
simplified to\vspace{-8pt}
\begin{align}\label{E49}
&\text{E}_{\Sigma}\!\!\left[\!\sum_{m=1}^{K}(2^{\frac{B_{\rm{max}}}{b}}\!\!-\!\!1)(E_1(a_{1,m})\!\!-\!\!E_1(a_{2,m}))\!\!+\!\!
\sqrt[m+bK]{\sigma^2_{s_1}...\sigma^2_{s_m}}\!\!\sqrt[m+bK]{\!\!\left(\frac{bK}{\lambda}\!\right)^{m}}\left(\!\Gamma\!\left(\frac{-bK}{m+bK}+1,e_{1,m}\!\right)\!\!-\!\!\right.\right.\nonumber\\
&\left.\left.\Gamma\left(\!\frac{-bK}{m+bK}+1,e_{2,m}\!\right)\!\right)\!+\!E_1(e_{2,m})\!-\!E_1(e_{1,m})\!\right]\!\!=\!\!\bar{P}
\end{align}
\end{proposition}
\begin{remark}
According to Proposition \ref{P1}, if
$\alpha\leq\frac{\lambda}{bK\sigma_1^2}$ (case 1 in \eqref{E25} and
\eqref{E26}), channel transmission power and all the source rates in
a given frame are set to zero. This corresponds to the scenario in
which the channel gain and/or the variances of the source blocks
within a frame are relatively small. Now suppose that
$d_{1,m}\leq\alpha< d_{2,m}$, where $m\in\{1,...,K\}$ and $d_{1,m}$
and $d_{2,m}$ are functions of the source variances in states
$\Sigma$. For $\alpha< c_m$ (case 2 in \eqref{E25} and \eqref{E26}),
the available channel rate, which equals the channel capacity, is
allocated to the $m$ source blocks with highest variances within the
$K$ source blocks of the frame. For $\alpha\geq c_m$ (case 3 in
\eqref{E25} and \eqref{E26}), due to the buffer size constraint, the
channel coding rate is set to $B_{\rm{max}}/b$ and is allocated to
the $n$ source blocks with highest variances within the $K$ source
blocks of the frame, where $n$ is given in Proposition \ref{P1}.
Clearly, in case 2 the instantaneous channel capacity of the frame
constrains the rate allocation to source blocks within the frame. In
case 3, the buffer size constrains the rate allocation instead.
\end{remark}
\begin{proof}
The proof is provided in Appendix \ref{A2}.
\end{proof}

The next two propositions quantify the performance of the proposed
SCORPA scheme in terms of the mean distortion and mean distortion
exponent, respectively.
\begin{proposition}\label{P4}
Mean distortion obtained by SCORPA scheme for transmission of a
quasi-stationary source over a Rayleigh block fading channel is
given by\vspace{-2pt}
\begin{align}\label{E43}
&\text{E}[D]\!\!=\!\!\frac{1}{K}\text{E}_{\Sigma}\left[\sum_{j=1}^{K}\sigma^2_{s_j}\left(1\!-\!\exp(-{\frac{\lambda}{bK\sigma^2_{s_1}}})\right)\!\!+\!\!
\sum_{m=1}^{K}\!\left(\!n\mathbf{\tilde{\lambda}}_2+\sum_{j=n+1}^{K}{\sigma^2_{s_j}}\!\right)\!\!\bigl(\exp(-a_{1,m})\!\!-\!\!\exp(-a_{2,m})\bigr)\!+\!\!\right.\nonumber\\
&\left.m\sqrt[m+bK]{\sigma^2_{s_1}...\sigma^2_{s_m}}\sqrt[m+bK]{\left(\frac{\lambda}{bK}\right)^{Kb}}\left(\Gamma(\frac{-bK}{m+bK}+1,e_{1,m})-\Gamma(\!\frac{-bK}{m+bK}+1,e_{2,m})\!\right)\!+\!\right.\nonumber\\
&\left.\sum_{j=m+1}^{K}\!\sigma^2_{s_j}\left(\exp(-e_{1,m})\!-\!\exp(-e_{2,m})\right)\right],
\end{align}
where $\lambda$, $\mathbf{\tilde{\lambda}}_2$, $n$, $a_{1,m}$,
$a_{2,m}$, $e_{1,m}$ and $e_{2,m}$ are defined in Proposition
\ref{P1}.
\end{proposition}
\begin{proof}
Considering exponentially distributed channel gain, using
Proposition \ref{P1}, \eqref{E1}-\eqref{E137}, achieving \eqref{E43}
is straightforward.
\end{proof} 
\begin{proposition}\label{P8}
For a Rayleigh block fading channel and with large power $\bar{P}$,
SCORPA scheme achieves the mean distortion exponents
$\Delta^{B_{\rm{max}}}_{\textit{MD}}=0$ and
$\Delta_{\textit{MD}}=b$, respectively.
\end{proposition}
\begin{proof}
We first consider the buffer constrained scenario. As evident from
\eqref{E48}, we have $\lambda\rightarrow 0$ for large power
constraint $\bar{P}\rightarrow \infty$. From \eqref{E100}, we
obtain\vspace{-12pt}
\begin{align}\label{E37}
e_{i,m}\rightarrow 0\;\;\;\forall{m}\in
\{1,...,K-1\}\;\;\;i\in\{1,2\},\;\;\;\;\;\;\;
e_{1,K}\rightarrow 0,\;\;\;\;\;\;\;e_{2,K}\rightarrow \infty.
\end{align}
Thus, noting \eqref{E98}, we have\vspace{-6pt}
\begin{align}\label{E60}
&E_1(e_{1,m})-E_1(e_{2,m})=-\ln(\frac{e_{1,m}}{e_{2,m}}),\;\;\;\forall{m}\in \{1,...,K-1\}\\
&\Gamma\left(\frac{-bK}{m+bK}+1,e_{1,m}\right)-\Gamma\left(\frac{-bK}{m+bK}+1,e_{1,m}\right)=0,\;\;\;\forall{m}\in
\{1,...,K-1\}\\
&E_1(e_{1,K})-E_1(e_{2,K})=-\ln(e_{1,K})\\
&\Gamma\left(\frac{-bK}{m+bK}+1,e_{1,K}\right)-\Gamma\left(\frac{-bK}{m+bK}+1,e_{2,K}\right)=\Gamma\left(\frac{-bK}{m+bK}+1,0\right)\label{E61}.
\end{align}
From \eqref{E100} it is observed that the equations similar to
\eqref{E37} to \eqref{E61} are satisfied for $a_{1,m}$ and
$a_{2,m}$. Hence, from \eqref{E49} and by ignoring
$\ln(\frac{bK}{\lambda})$ in contrast to $\frac{bK}{\lambda}$ for
$\lambda\rightarrow0$, we obtain\vspace{-2pt}
\begin{align}\label{E36}
\text{E}_{\Sigma}\left[
\sqrt[K+bK]{\sigma^2_{s_1}...\sigma^2_{s_K}}\sqrt[K+bK]{\left(\frac{bK}{\lambda}\right)^{K}}\Gamma\left(\frac{-bK}{K+bK}+1,0\right)
\right]=\bar{P}.
\end{align}
Or equivalently\vspace{-4pt}
$\sqrt[1+b]{\left(\frac{bK}{\lambda}\right)}\Gamma\left(\frac{1}{b+1},0\right)\text{E}_{\Sigma}\left[
\sqrt[K+bK]{\sigma^2_{s_1}...\sigma^2_{s_K}} \right]=\bar{P}.$
Thus,
\begin{align}\label{E50}
\sqrt[1+b]{\lambda}=\frac{\sqrt[1+b]{\left(bK\right)}\Gamma\left(\frac{1}{b+1},0\right)\text{E}_{\Sigma}\left[
\sqrt[K+bK]{\sigma^2_{s_1}...\sigma^2_{s_K}} \right]}{\bar{P}}.
\end{align}
From \eqref{E37} to \eqref{E61} and noting $\lambda\rightarrow 0$,
only the second term in \eqref{E43} is non-negligible and therefore
mean distortion is given by\vspace{-10pt}
\begin{align}\label{E38}
&\text{E}[D]=\frac{1}{K}\text{E}_{\Sigma}\left[n\mathbf{\tilde{\lambda}}_2+\sum_{j=n+1}^{K}{\sigma^2_{s_j}}\right]=
\frac{1}{K}\text{E}_{\Sigma}\left[n\sqrt[n]{\frac{\sigma_{s_1}^2...\sigma^2_{s_n}}{2^{KB_{\rm{max}}}}}+\sum_{j=n+1}^{K}{\sigma^2_{s_j}}\right],
\end{align}
where $n$ is an integer in $\{1,2,...,K\}$ such that
$\sigma^2_{s_{n+1}}<\sqrt[n]{\frac{\sigma_{s_1}^2...\sigma^2_{s_n}}{2^{KB_{\rm{max}}}}}\leq\sigma^2_{s_{n}}$.
Therefore for any finite buffer constraint $B_{max}$,
$\Delta_{\textit{MD}}^{B_{\rm{max}}}=0.$

We now consider the buffer unconstrained scenario. In this case we
first let $B_{max} \rightarrow \infty$ in Propositions 1 and 2 and
then increase the power $\bar{P}$. For $B_{\rm{max}}\rightarrow
\infty$ and noting \eqref{E52}, we have $c_m\rightarrow\infty$.
Thus, from \eqref{E100}, we obtain
$a_{1,m}=a_{2,m}=c_m\rightarrow\infty$. Hence, the first term in
\eqref{E48} is omitted and for large power constraint
$\lambda\rightarrow 0$. With $\lambda\rightarrow 0$, we still obtain
\eqref{E37} to \eqref{E50}. The second term in \eqref{E43} is
omitted and therefore, the mean distortion in \eqref{E43} is
simplified as
\begin{align}\vspace{-8pt}
&\text{E}[D]\!\!=\!\!\frac{1}{K}\text{E}_{\Sigma}\!\!\left[\!\!\sum_{j=1}^{K}\sigma^2_{s_i}\frac{\lambda}{bK\sigma^2_{s_1}}\!\!+\!\!
K\sqrt[K+bK]{\sigma^2_{s_1}...\sigma^2_{s_K}}\sqrt[1+b]{\!\left(\frac{\lambda}{bK}\!\right)^{b}}\Gamma\!\left(\!\frac{1}{b+1},\!\right)\!\!+\!\!
\sum_{m=1}^{K}\sum_{j=m+1}^{K}\!\sigma^2_{s_i}\!\left(d_{2,m}\!-\!d_{1,m}\right)\!\!\right]\!,\nonumber
\end{align}
where noting $d_{1,m}$ and $d_{2,m}$ in \eqref{E53}, ignoring
$\frac{\lambda}{bK}$ with respect to
$\left(\frac{\lambda}{bK}\right)^{\frac{b}{b+1}}$ for
$\lambda\rightarrow 0$ and utilizing \eqref{E50} we
have\vspace{-12pt}
\begin{align}\label{E55}
\text{E}[D]=\frac{1}{\bar{P}^b}\left(\text{E}_{\Sigma}\left[
\sqrt[K+bK]{\sigma^2_{s_1}...\sigma^2_{s_K}}\right]\right)^{b+1}\Gamma\left(\frac{1}{b+1},0\right)^{b+1}.
\end{align}
Therefore for any finite buffer constraint $B_{max}$,
$\Delta_{\textit{MD}}=b$.
Thus, the proof is complete.
\end{proof}
The performance of the proposed SCORPA scheme is studied in Section
\ref{SVI}.

\section{Channel Optimized Power Adaptation with Constant
Channel Coding Rate} \label{SIII}
 In this section, the aim is to
find the optimized power allocation strategy and non-adaptive
channel code rate $R$ such that the mean distortion for
communication of a delay-limited quasi-stationary source over a
block fading channel minimized. With a fixed channel rate, the
channel code can be fixed, which simplifies the design and
implementation of transceivers. Note that source coding rates in
different blocks of a frame, $R_{s_j}$, may still be adapted. Also,
we will find the best fixed channel coding rate $R$ to minimize the
expected source distortion. The mean distortion in \eqref{E1} for a
fixed channel rate $R$ is simplified as follows\vspace{-2pt}
\begin{align}\label{E22}
\text{E}[D]=\text{E}_{\Sigma}\left[\frac{\sum_{j=1}^{K}\sigma^2_{s_j}2^{-R_{s_j}}}{K}\right]\bigl(1-\text{Pr}\bigl(R>
C\left(\alpha,\gamma\right)\bigr)\bigr)+\text{E}_s[\sigma_s^2]
\text{Pr}\bigl(R> C\left(\alpha,\gamma\right)\bigr),
\end{align}
which is subject to the constraints in \eqref{E47} and \eqref{E45}.
We have the following design problem.
\begin{problem}\label{PR4}
The problem of delay-limited Channel Optimized Power Adaptation with
Constant Channel Rate (COPACR) for communication of a
quasi-stationary source with limited buffer over a block fading
channel is formulated as follows\vspace{-8pt}
\begin{align}\label{E31}
&\underset{\gamma(\Sigma,\alpha),R,R_{s_j}(\Sigma,\alpha)}{\operatorname{min}}\;\text{E}[D] \;\;\;\text{subject to}\nonumber\\
&
\text{E}[\gamma]\leq\bar{P},\;\;\;\;\;\sum_{j=1}^{K}R_{s_j}=KbR,\;\;\;\;\;\;bR\leq
B_{\rm{max}},
\end{align}
\end{problem}\vspace{-4pt}
where $\text{E}[D]$ is given in \eqref{E22}.

The solution to Problem \ref{PR4} is obtained in three steps. We may
rewrite \eqref{E22} as follows\vspace{-2pt}
\begin{align}\label{E32}
\text{E}[D]=\text{Pr}\bigl(R>
C\left(\alpha,\gamma\right)\bigr)\biggl(\text{E}_s[\sigma_s^2]-\text{E}_{\Sigma}\left[\frac{\sum_{j=1}^{K}\sigma^2_{s_j}2^{-R_{s_j}}}{K}\right]
\biggr)+\text{E}_{\Sigma}\left[\frac{\sum_{j=1}^{K}\sigma^2_{s_j}2^{-R_{s_j}}}{K}\right].
\end{align}
For a given channel coding rate $R$, power adaptation only affects
the mean distortion in \eqref{E32} through the term
$\text{Pr}\bigl(R> C\left(\alpha,\gamma\right)\bigr)$. Thus, the
optimum $\gamma$ does not depend on $\Sigma$ (vector of source
states). We first consider the design (sub)problem \ref{PR7} below
to find the optimized power adaptation strategy for a given $R$.
Next, assuming a given $R$ and the resulting power adaptation
strategy, noting \eqref{E22}, we consider another design
(sub)problem \ref{PR8} for source rate allocation to different
blocks over a frame which aims at minimizing the term
$\text{E}_{\Sigma}\left[{\sum_{j=1}^{K}\sigma^2_{s_j}2^{-R_{s_j}}}\right]$.
As we shall see, the results of the two (sub)problems are derived in
terms of $R$, which is then directly obtained by solving problem
\ref{PR4}.

Problem \ref{PR7} below formulates the power adaptation strategy for
a given $R$. Note that\vspace{-2pt}
\begin{equation}
\text{E}_{\Sigma}\left[\frac{\sum_{j=1}^{K}\sigma^2_{s_j}2^{-R_{s_j}}}{K}\right]\leq
\text{E}_s[\sigma_s^2],\; \forall R_{s_j}\ge0.
\end{equation}
\begin{problem}\label{PR7}
With COPACR scheme and with a given channel coding rate $R$, the
power adaptation problem is formulated as follows\vspace{-8pt}
\begin{align}\label{E24}
\underset{\gamma}{\operatorname{min}}\;\text{Pr}(R>C(\alpha,\gamma))
\;\text{subject to}\;\;\;\;\;\text{E}[\gamma]\leq\bar{P}{.}
\end{align}
\end{problem}
\begin{proposition}\label{P6}
Solution to Problem \ref{PR7} for optimized power adaptation over a
block fading channel is given by\vspace{-4pt}
\begin{equation}\label{E28}
\gamma^*{(\alpha,R)}=\begin{cases}
\frac{2^{R}-1}{\alpha} &\text{if}\;\;\alpha\geq\frac{2^{R}-1}{q^*_1(R)} \\
0 &\text{if}\;\;\alpha<\frac{2^{R}-1}{q^*_1(R)}{,}
\end{cases}
\end{equation}
in which $q^*_1(R)$ is selected such that the power constraint is
satisfied with equality, i.e.,\vspace{-6pt}
\begin{equation}\label{E129}
\underset{\alpha\geq\frac{2^{R}-1}{q^*_1(R)}}{\int}{\frac{2^{R}-1}{\alpha}
f(\alpha)\,\mathrm{d}\alpha}=\bar{P}.
\end{equation}\vspace{-4pt}
Specifically for Rayleigh block fading channel, \eqref{E129} is
simplified to\vspace{-2pt}
\begin{equation}\label{E54}
\left(2^{R}-1\right)E_1\left(\frac{2^{R}-1}{q^*_1(R)}\right)=\bar{P}.
\end{equation}
\end{proposition}
\begin{proof}
The proof follows that of Proposition 4 in \cite{R1}. Based on
\eqref{E129}, for Rayleigh block fading channel $q_1^*(R)$ is to
satisfy \eqref{E54}. In fact $q_1^*(R)$ sets the SNR threshold below
which the channel outage occurs.
\end{proof}

For a given channel coding rate $R$ and the optimized power
adaptation strategy in Proposition \ref{P6}, the outage probability
$\bigl(R> C\left(\alpha,\gamma\right)\bigr)$ is fixed. Hence,
minimizing mean distortion in \eqref{E22} is equivalent to
minimizing
$\text{E}_{\Sigma}\left[{\sum_{j=1}^{K}\sigma^2_{s_j}2^{-R_{s_j}}}\right]$.
This leads us to Problem \ref{PR8} below which formulates the source
rate allocation design for a given $R$ within the COPACR scheme.
\begin{problem}\label{PR8}
For COPACR scheme with a given $R\leq\frac{B_{\rm{max}}}{b}$, the
set of optimum source coding rates, $R^*_{s_j}, j\in\{1,...,K\}$, is
given by the solution to the following optimization
problem\vspace{-6pt}
\begin{equation}\label{E8}
\begin{split}
&\underset{R_{s_j}}{\operatorname{min}}\;\text{E}_{\Sigma}\left[{\sum_{j=1}^{K}\sigma^2_{s_j}2^{-R_{s_j}}}\right]\;\;\text{subject
to:
}\\
& \sum_{j=1}^{K}R_{s_j}= KbR,\;\;\;\;\;\;\;R_{s_j}\geq 0.
\end{split}
\end{equation}
\end{problem}
\begin{proposition}\label{P7}
The solution to Problem \ref{PR8} for an arbitrary block fading
channel is given by\vspace{-6pt}
\begin{equation}
\begin{split}\label{E82}
R^*_{s_j}=&\left[\log_2\frac{\sigma_{s_j}^2}{\lambda(R)}\right]^+
\end{split}
\end{equation}
with
$\lambda(R)=\sqrt[n]{\frac{\sigma^2_{s_1}\times...\times\sigma^2_{s_n}}{2^{bKR}}},$
where $n$ is an integer in $\{1,2,...,K\}$ such that
$\sigma^2_{s_{n+1}}<\lambda(R)\leq\sigma^2_{s_{n}}.$
\end{proposition}
\begin{proof}
Using reverse water filling \cite{R3}, the proof is straightforward.
\end{proof}

Now Proposition \ref{P10} below gives solution to Problem \ref{PR4}.
\begin{proposition}\label{P10}
The solution to Problem \ref{PR4} for an arbitrary block fading
channel is given by\vspace{-4pt}
\begin{equation}\label{E10}
\begin{split}
R^*\!\!=\!\!\underset{R:\;0\leq bR\leq
B_{\rm{max}}}{\operatorname{argmin}}\;\!\!\text{E}[D]\!\!=\!\!\text{E}_s[\sigma_s^2]\text{Pr}\biggl(\alpha<\frac{2^{R}-1}{q^*_1(R)}\biggr)\!\!+\!\!
\biggl(\!1\!-\!\text{Pr}\biggl(\!\alpha<\frac{2^{R}-1}{q^*_1(R)}\!\biggr)\!\biggr)\text{E}_{\Sigma}\!\!\left[\!\frac{n\lambda(R)\!\!+\!\!\sum_{j=n+1}^{K}\sigma^2_{s_j}}{K}\!\right]
\end{split}
\end{equation}
with $q^*_1(R)$ satisfying \eqref{E129} and $\lambda(R)$ given in
Proposition \ref{P7}. Consequently with the obtained $R^*$,
$\gamma^*{(\alpha,R^*)}$ and $R^*_{s_j}$ may be derived using
Propositions \ref{P6} and \ref{P7}, respectively.
\end{proposition}

\begin{proof}
Using Propositions \ref{P6} and \ref{P7}, we replace the optimized
power adaptation and source rate allocation results into
\eqref{E22}, which provides the mean distortion in terms of a single
variable $R$. Hence, the optimized value of $R$ may be obtained
numerically as indicated in \eqref{E1}. Thus, the proof is complete.
Our extensive studies in the case of Rayleigh block fading channel
reveals that the mean distortion in \eqref{E10} is a convex function
of $R$ and hence indicates a unique minimum at $R^*$.
\end{proof}

Proposition \ref{P14} below provides the mean distortion achieved by
COPACR.
\begin{proposition}\label{P14}
For transmission of a quasi-stationary source over a Rayleigh block
fading channel, the COPACR scheme achieves the mean distortion
\begin{align}\label{E81}
\text{E}[D]=\text{E}_s[\sigma_s^2]\left(1-\exp(-\frac{2^{R^*}-1}{q^*_1(R^*)})\right)+\exp\left(-\frac{2^{R^*}-1}{q^*_1(R^*)}\right)\text{E}_{\Sigma}\left[\frac{n\lambda(R^*)+\sum_{j=n+1}^{K}\sigma^2_{s_j}}{K}\right]
\end{align}
 with $n$, $R^*$ and $\lambda(R^*)$ in Proposition \ref{P7}
and
\begin{equation}\label{E91}
\left(2^{R^*}-1\right)E_1\left(\frac{2^{R^*}-1}{q^*_1(R^*)}\right)=\bar{P}.
\end{equation}
\end{proposition}

\begin{proof}
Using Proposition \ref{P10} for Rayleigh block fading channel,
achieving \eqref{E81} and \eqref{E91} is straightforward.
\end{proof}

\begin{proposition}\label{P15}
For a Rayleigh block fading channel, the COPACR scheme achieves the
mean distortion exponents $\Delta^{B_{\rm{max}}}_{\textit{MD}}=0$
and $\Delta_{\textit{MD}}=br_1$, where $r_1\in[0,1)$ denotes the
COPACR buffer unconstrained multiplexing gain which can be
calculated by numerically approximating $R^*$ as
$R^*={r}_1\log_2\bar{P}+{r}_0$.
\end{proposition}

\begin{proof}
The proof is provided in Appendix \ref{A3}.
\end{proof}
The performance of COPACR scheme is studied and compared in Section
\ref{SVI}.

\section{Constant Power schemes}\label{SV}
In this section, two constant power schemes for delay-limited
transmission of a quasi-stationary source over a block fading
channel with buffer and power limitations are considered. In the
first scheme, the source and channel coding rates are adjusted based
on the source or channel states to minimize the mean distortion;
hence the scheme is labeled as Source and Channel Optimized Rate
Adaptation with Constant Power (SCORACP). In the second scheme with
Constant Rate and Constant Power (CRCP), the aim is to find the
optimized fixed channel rate and adaptive source coding rates such
that the mean distortion is minimized.

\subsection{Source and Channel Optimized Rate Adaptation with
Constant Power} With SCORACP and constant transmission power
$\bar{P}$, the instantaneous capacity is
$C(\alpha)=\log_2{\left(1+\alpha\bar{P}\right)}$. As discussed, the
source and channel coding rates may be controlled with respect to
the channel state such that $R\leq C(\alpha)$. Thus, noting the
buffer size constraint $B_{\rm{max}}$, we have the next design
problem.

\begin{problem}\label{PR3}
The problem of SCORACP for delay-limited communication of a
quasi-stationary source with limited buffer over a block fading
channel is formulated as follows for minimum mean
distortion,\vspace{-12pt}
\begin{align}\label{E17}
&\underset{{R}_{sj}(\Sigma,\alpha)}{\operatorname{min}}\;{\text{E}_{\Sigma,\alpha}\left[{\sum_{j=1}^{K}\sigma^2_{s_j}2^{-R_{s_j}}}\right]}\;\;\text{   subject to}\\
&\;{\sum_{j=1}^{K}R_{s_j}}\leq K B_{\rm{max}},\;\;\;\;\;
{\sum_{j=1}^{K}R_{s_j}}\leq K b C(\alpha)\label{E120}.
\end{align}
\end{problem}

\begin{proposition}\label{P2}
Solution to Problem \ref{PR3}, denoted by $R^*_{s_j},
j\in\{1,...,K\}$, for an arbitrary block fading channel is given
by\vspace{-4pt}
\begin{align}\vspace{-8pt}
R^*_{s_j}=\begin{cases}\left[\log_2\frac{\sigma_{s_j}^2}{\mathbf{\tilde{\lambda}}}\right]^+&\text{if}\;\alpha\geq
c\\
\left[\log_2\frac{\sigma_{s_j}^2}{\lambda}\right]^+&\text{if}\;\alpha<
c,\;\;\;d_{1,m}\leq\alpha<d_{2,m},\forall m\in{\{1,...,K\}},
\end{cases}
\end{align}
where $
\lambda=\sqrt[m]{\frac{\sigma_{s_1}^2...\sigma^2_{s_m}}{(1+\alpha\bar{P})^{bK}}},\;\mathbf{\tilde{\lambda}}=\sqrt[n]{\frac{\sigma_{s_1}^2...\sigma^2_{s_n}}{2^{KB_{\rm{max}}}}}$
and $n$ is an integer in $\{1,2,...,K\}$ such that
$\sigma^2_{s_{n+1}}<\mathbf{\tilde{\lambda}}\leq\sigma^2_{s_{n}}$;
and\vspace{-8pt}
\begin{align}\label{E57}
d_{1,m}=\frac{\frac{\sqrt[bK]{\sigma^2_{s_1}...\sigma^2_{s_m}}}{(\sigma^2_{s_m})^{\frac{m}{bK}}}-1}{\bar{P}},\;\;\;\;\;
 d_{2,m}=\begin{cases}\frac{\frac{\sqrt[bK]{\sigma^2_{s_1}...\sigma^2_{s_m}}}{(\sigma^2_{s_{m+1}})^{\frac{m}{bK}}}-1}{\bar{P}}&\text{if}\;m\neq K\\
\infty&\text{if} \;m=K \end{cases},\;\;\;\;\;
c=\frac{2^{\frac{B_{\rm{max}}}{b}}-1}{\bar{P}}.
\end{align}
\end{proposition}
\begin{proof}
The proof is provided in Appendix \ref{A4}.
\end{proof}

The next two propositions quantify the mean distortion performance
of SCORACP.
\begin{proposition}\label{P4}
The mean distortion obtained by SCORACP for transmission of a
quasi-stationary source over a Rayleigh block fading channel is
given by\vspace{-4pt}
\begin{align}\label{E42}
&\text{E}[D]=\frac{1}{K}\text{E}_{\Sigma}\left[\left(n\mathbf{\tilde{\lambda}}+\sum_{j=n+1}^{K}\sigma^2_{s_j}\right)\exp(-c)+\sum_{m=1}^{K}
\left\{m\!\sqrt[m]{\sigma_{s_1}^2...\sigma^2_{s_m}}\bar{P}^{-1}\exp(\!\frac{1}{\bar{P}}\!)\right.\right.\nonumber\\
&\left.\left.\left((1+d_{1,m}\bar{P}\!)^{-\frac{bK}{m}+1}\!E_{\frac{bK}{m}}(\frac{1}{\bar{P}}+d_{1,m}\!)
-(1+\min(d_{2,m},c)\bar{P})^{-\frac{bK}{m}+1}E_{\frac{bK}{m}}(\frac{1}{\bar{P}}+\min(d_{2,m},c))\right)+\nonumber\right.\right.\\
&\left.\left.\sum_{i=m+1}^{K}\!\sigma^2_{s_i}\left(\exp(-d_{1,m})-\exp(-\min(d_{2,m},c))\right)\right\}\right],
\end{align}
where $\mathbf{\tilde{\lambda}}$, $n$, $d_{1,m}$ and $d_{2,m}$ are
defined in Proposition \ref{P2}.
\end{proposition}

\begin{proof}
Noting \eqref{E1} and Proposition \ref{P2}, we obtain
\begin{align}\label{E56}
\text{E}[D]\!\!=\!\!\frac{1}{K}\text{E}_{\Sigma}\!\!\left[\!\!\int_{\alpha>c}\!\left(\!\!n\mathbf{\tilde{\lambda}}\!\!\!+\!\!\!\!\sum_{j=n+1}^{K}\sigma^2_{s_j}\!\!\right)f(\alpha)d\alpha\!\!+\!\sum_{m=1}^{K}
\!\int_{\!\!\alpha\leq
c,\;d_{1,m}\leq\alpha<d_{2,m}}\!\!\!\left(\!m\sqrt[m]{\frac{\sigma_{s_1}^2...\sigma^2_{s_m}}{\!(1\!+\!\!\alpha\bar{P}\!)^{bK}}}\!\!+\!\!\!\!\sum_{i=m+1}^{K}\!\sigma^2_{s_i}\!\right)\!f(\alpha)d\alpha\!\!\right]\!.
\end{align}
We change the variables as $u_1=\frac{1+\alpha\bar{P}}{\bar{P}}$ and
$u_2=1+\alpha\bar{P}$, respectively for $\frac{bK}{m}=1$ and
$\frac{bK}{m}>1$ in the second integral. Hence, noting the
exponentially distributed channel gain and using \eqref{E137},
equation \eqref{E56} is rewritten as \eqref{E42}.
\end{proof}

\begin{proposition}\label{P16}
For a Rayleigh block fading channel and with large power $\bar{P}$,
the SCORACP scheme achieves the mean distortion exponents
$\Delta^{B_{\rm{max}}}_{\textit{MD}}=0$ and $\Delta_{\textit{MD}}=1$
\end{proposition}

\begin{proof}
We first consider the buffer constrained scenario. As denoted, it
can intuitively be seen that $\Delta^{B_{\rm{max}}}_{\textit{MD}}$
of SCORACP is zero. However, because we need the mean distortion for
large power constraint in Section \ref{SVIB} we bring here the
steps. For large power $\bar{P}\rightarrow \infty$ and from
\eqref{E57}, we have\vspace{-12pt}
\begin{align}\label{E93}
c\rightarrow 0,\;\;\;\;\;d_{1,m} \rightarrow 0 \;\;\;\forall
m\in\{1,...,K\},\;\;\;\;\;d_{2,m}\rightarrow\begin{cases} 0 &\forall
m\in\{1,...,K-1\}\\
\infty &m=K.
\end{cases}
\end{align}
Note that for $b>1$ and $m\in\{1,...,K\}$, $\frac{bK}{m}$ is greater
than 1. Thus for $\bar{P}\rightarrow\infty$ and $b>1$, utilizing
\eqref{E98}, we have\vspace{-12pt}
\begin{align}
E_{\frac{bK}{m}}(\frac{1}{\bar{P}}+d_{1,m}\!)=\begin{cases}\frac{1}{\frac{bK}{m}-1}
&m\in\{1,...,K-1\}\text{ and }b=1\\
E_{1}(\frac{1}{\bar{P}}+d_{1,m}\!) & m=K \text{ and }b=1\\
\frac{1}{\frac{bK}{m}-1} &b>1 \end{cases}
\end{align}
and $\forall b>1$\vspace{-12pt}
\begin{align}
E_{\frac{bK}{m}}(\frac{1}{\bar{P}}+d_{2,m})=\begin{cases}\frac{1}{\frac{bK}{m}-1}&\forall
m\in\{1,...,K-1\}\\
0&m=K.
\end{cases}
\end{align}\vspace{-6pt}
Hence, from \eqref{E42} we obtain\vspace{-6pt}
\begin{align}\label{E62}
&\text{E[D]}=\frac{1}{K}\text{E}_{\Sigma}\left(n\tilde{\lambda}+\sum_{j=n+1}^{K}\sigma^2_{s_j}\right)\exp(-c)+\bar{P}^{-1}V_m,\;\;b>1,
\end{align}
where\vspace{-8pt}
\begin{align}\label{E79}
&V_m=\frac{1}{K}\text{E}_{\Sigma}\left[\sum_{m=1}^{K-1}\left\{m\!\sqrt[m]{\sigma_{s_1}^2...\sigma^2_{s_m}}\left((x_{1,m})^{-\frac{bK}{m}+1}-(x_{2,m})^{-\frac{bK}{m}+1}\right)\frac{1}{\frac{bK}{m}-1}\right.\right.\nonumber\\
&\left.\left.+\left(x_{2,m}-x_{1,m}\right)\sum_{i=m+1}^{K}\sigma^2_{s_i}\right\}+
K\!\sqrt[K]{\sigma_{s_1}^2...\sigma^2_{s_K}}(x_{1,K})^{-b+1}\frac{1}{b-1}
\right],
\end{align}
\begin{align}\label{E80}
x_{i,m}=\frac{\sqrt[bK]{\sigma^2_{s_1}...\sigma^2_{s_m}}}{(\sigma^2_{s_{m+i-1}})^{\frac{m}{bK}}},\;i=\{1,2\},
\end{align}\vspace{-10pt}
and for $b=1$, we have
\begin{align}\label{E78}
\text{E[D]}=\frac{1}{K}\text{E}_{\Sigma}\left(n\tilde{\lambda}+\sum_{j=n+1}^{K}\sigma^2_{s_j}\right)\exp(-c)+\bar{P}^{-1}W_m,\;\;b=1,
\end{align}
where\vspace{-8pt}
\begin{align}\label{E104}
&W_m=\frac{1}{K}\text{E}_{\Sigma}\left[\sum_{m=1}^{K-1}\left\{m\!\sqrt[m]{\sigma_{s_1}^2...\sigma^2_{s_m}}\left((x_{1,m})^{-\frac{bK}{m}+1}-(x_{2,m})^{-\frac{bK}{m}+1}\right)\frac{1}{\frac{bK}{m}-1}\right.\right.\nonumber\\
&\left.\left.+\left(x_{2,m}-x_{1,m}\right)\sum_{i=m+1}^{K}\sigma^2_{s_i}\right\}+
KE_1\left(\frac{x_{1,K}}{\bar{P}}\right)\sqrt[K]{\sigma_{s_1}^2...\sigma^2_{s_K}}
\right].
\end{align}
Using \eqref{E98}, it is evident that the second term may be
neglected with respect to the first term in both \eqref{E62} and
\eqref{E78} for $\bar{P}\rightarrow \infty$ and limited buffer size;
and therefore, we obtain\vspace{-6pt}
\begin{align}\label{E40}
&\text{E}[D]=\frac{1}{K}\text{E}_{\Sigma}\left[n\tilde{\lambda}+\sum_{j=n+1}^{K}{\sigma^2_{s_j}}\right]=
\frac{1}{K}\text{E}_{\Sigma}\left[n\sqrt[n]{\frac{\sigma_{s_1}^2...\sigma^2_{s_n}}{2^{KB_{\rm{max}}}}}+\sum_{j=n+1}^{K}{\sigma^2_{s_j}}\right],
\end{align}
where $n$ is an integer in $\{1,2,...,K\}$ such that
$\sigma^2_{s_{n+1}}<\sqrt[n]{\frac{\sigma_{s_1}^2...\sigma^2_{s_n}}{2^{KB_{\rm{max}}}}}\leq\sigma^2_{s_{n}}$.
Therefore, $\Delta_{\textit{MD}}^{B_{\rm{max}}}=0.$

We now consider the buffer unconstrained scenario. Obviously from
\eqref{E57}, we have $c\rightarrow\infty$. However, we still have
\eqref{E93} to \eqref{E104} and the first term in \eqref{E62} and
\eqref{E78} is zero. Therefore, using \eqref{E98} in \eqref{E62} and
\eqref{E78}, we derive
$\Delta_{\textit{MD}}=\underset{\bar{P}\rightarrow\infty}{\lim}-\frac{\ln{\text{E}[D]}}{\ln{\bar{P}}}=1.$
Thus, the proof is complete.
\end{proof}\vspace{-14pt}
\subsection{Constant Channel Rate with Constant Power} With CRCP scheme, the
fixed channel rate and adaptive source coding rates are optimized to
minimize the mean distortion when the power is constant. The
following three propositions express the optimized rate, mean
distortion and mean distortion exponent obtained by CRCP.

\begin{proposition}\label{P12}
For delay-limited transmission of a quasi-stationary source over a
block fading channel using the CRCP scheme and limited buffer, the
optimum rates $R^*$ and $R^*_{s_j},\;j\in\{1,...,K\}$ for minimum
mean distortion are obtained as follows,
\begin{align}
R^*=&\underset{R:\;0\leq bR\leq
B_{\rm{max}}}{\operatorname{argmin}}\!\!\!\text{E}[D]\!\!=\!\!\text{E}_s[\sigma_s^2]\text{Pr}\biggl(\!\alpha\!\!<\!\!\frac{2^{R}-1}{\bar{P}}\!\biggr)\!\!+\!\!
\biggl(1\!\!-\!\!\text{Pr}\biggl(\alpha<\frac{2^{R}-1}{\bar{P}}\biggr)\biggr)\!\text{E}_{\Sigma}\!\left[\!\!\frac{n\lambda(R)\!+\!\sum_{j=n+1}^{K}\sigma^2_{s_j}}{K}\!\!\right]\!,\label{E123}\\
R^*_{s_j}=&\left[\!\!\log_2\frac{\sigma_{s_j}^2}{\lambda(R^*)}\right]^+,\text{
where}\;\;\;\lambda(R)=\sqrt[n]{\frac{\sigma^2_{s_1}\times...\times\sigma^2_{s_n}}{2^{bKR}}},\label{E71}
\end{align}
$n$ is an integer in $\{1,...,K\}$ such that
$\sigma_{s_i}^2\geq{\lambda}$ for $i\leq n$ and
$\sigma^2_{s_i}<{\lambda}$ for $i>n;\;i\in\{1,...,K\}$.
\end{proposition}
\begin{proof}
Since the power is constant, the channel capacity and the mean
distortion respectively are given as
$C(\alpha)=\log_2(1+\alpha\bar{P})$
and\vspace{-8pt}
\begin{align}\label{E122}
\text{E}[D]=\text{E}_s[\sigma_s^2]\text{Pr}\biggl(\alpha<\frac{2^{R}-1}{\bar{P}}\biggr)+
\text{E}_{\Sigma}\left[{\sum_{j=1}^{K}\sigma^2_{s_j}2^{-R_{s_j}}}\right]\biggl(1-\text{Pr}\biggl(\alpha<\frac{2^{R}-1}{\bar{P}}\biggr)\biggr).
\end{align}
For a given channel coding rate $R$, minimizing mean distortion in
\eqref{E122} is equivalent to minimizing
$\text{E}_{\Sigma}\left[{\sum_{j=1}^{K}\sigma^2_{s_j}2^{-R_{s_j}}}\right]$.
As noted, using reverse water filling, this leads us to Problem
\ref{PR8} and Proposition \ref{P7} which provides the source rate
allocation for a given $R$. Hence, the optimized value
of $R$ may be obtained numerically as indicated in \eqref{E123}.
Thus, the proof is complete.
\end{proof}

\begin{proposition}\label{P13}
For transmission of a quasi-stationary source over a Rayleigh block
fading channel, the CRCP scheme achieves the mean
distortion\vspace{-8pt}
\begin{align}\label{E74}
\text{E}[D]=\text{E}[\sigma^2]\left(1-\exp(\frac{2^{R^*}-1}{\bar{P}})\right)+\exp(\frac{2^{R^*}-1}{\bar{P}})\text{E}_{\Sigma}\left[\frac{n\lambda(R^*)+\sum_{j=n+1}^{K}\sigma^2_j}{K}\right]
\end{align}
with $n$, $R^*$ and $\lambda(R^*)$ in Proposition \ref{P12}, and the
mean distortion exponents $\Delta^{B_{\rm{max}}}_{\textit{MD}}=0$
and $\Delta_{\textit{MD}}=1-\tilde{r}_1$, where
$\tilde{r}_1=\frac{1}{b+1}$ denotes the CRCP buffer unconstrained
multiplexing gain.
\end{proposition}
\begin{proof}
The proof is provided in Appendix \ref{A5}.
\end{proof}

Note that although the approach used in the proof of Proposition
\ref{P13} is similar to that in \cite{R27}, the source considered in
\cite{R27} is stationary.

\section{Analytical Performance Comparison}\label{SVIB} In the
sequel, we quantify the respective asymptotic mean power gains
$G^{B_{\rm{max}}}_{\textit{MD}}$ and $G_{\textit{MD}}$ of SCORPA,
COPACR, SCORACP and CRCP for transmission of a quasi-stationary
source over a block fading channel. As observed from equations
\eqref{E4}, \eqref{E38}, \eqref{E40} and \eqref{E75} for large power
and with buffer constraints, all schemes provide the same mean
distortion independent of the power constraint. Since
$\Delta^{B_{\rm{max}}}_{\textit{MD}}=0$, there is no meaningful
$G^{B_{\rm{max}}}_{M\!D}$. Thus in the following we only address
$G_{\textit{MD}}$ and $\Delta_{\textit{MD}}$.

\begin{proposition}\label{P9}
In transmission of a quasi-stationary source over a Rayleigh block
fading channel, the asymptotic mean power gains obtained by SCORPA
with respect to COPACR, $G_{M\!D,1}$, COPACR with respect to
SCORACP, $G_{M\!D,2}$, and CRCP with respect to SCORACP,
$G_{M\!D,3}$, are given by\vspace{-4pt}
\begin{equation}\label{E68}
G_{M\!D,1}=\frac{10}{b}\log_{10}{\frac{\text{E}_{\Sigma}[\sqrt[K]{\sigma_{s_1}^2\times...\times\sigma_{s_K}^2}]\bar{P}_2^{b(1-r_1)}2^{-br_0}}
{(\Gamma(\frac{1}{b+1},0))^{b+1}(\text{E}_{\Sigma}[\sqrt[K+bK]{\sigma_{s_1}^2\times...\times\sigma_{s_K}^2}])^{b+1}}},
\end{equation}\vspace{-8pt}
\begin{align}\label{E69}
G_{M\!D,2}=\begin{cases}\frac{10}{r_1}\log_{10}\frac{W_m}{\text{E}_{\Sigma}[\sqrt[K]{\sigma_{s_1}^2\times...\times\sigma_{s_K}^2}]\bar{P}_2^{1-r_1}2^{-r_0}}&b=1\\
\frac{10}{br_1}\log_{10}{\frac{V_m}{\text{E}_{\Sigma}[\sqrt[K]{\sigma_{s_1}^2\times...\times\sigma_{s_K}^2}]\bar{P}_2^{1-br_1}2^{-br_0}}}&b>1\end{cases}
\end{align}
\vspace{-4pt}and\vspace{-8pt}
\begin{align}\label{E70}
G_{M\!D,3}=\begin{cases}20\log_{10}\!{\frac{W_m\bar{P}_2^{\frac{-1}{2}}}{\left(2^{-\tilde{r}_0}\text{E}_{\Sigma}\left[\sqrt[K]{\sigma^2_{s_1}\times...\times\sigma^2_{s_K}}\right]+2^{\tilde{r}_0}\text{E}_s[\sigma_s^2]\right)}},&b=1\\
\frac{b+1}{b}10\log_{10}\!{\frac{V_m\bar{P}_2^{\frac{-1}{b+1}}}{\left(2^{-b\tilde{r}_0}\text{E}_{\Sigma}\left[\sqrt[K]{\sigma^2_{s_1}\times...\times\sigma^2_{s_K}}\right]+2^{\tilde{r}_0}\text{E}_s[\sigma_s^2]\right)}},&b>1,\end{cases}
\end{align}
\end{proposition}
where $r_0$ and $r_1$ are given in Proposition \ref{P15},
$\tilde{r}_0$ is defined in Proposition \ref{P13}, $V_m$ and $W_m$
are respectively given in \eqref{E79} and \eqref{E104}; and
$\bar{P}_2$ is the power limit in scheme 2 (see Table \ref{T3} or
\ref{T4}).

\begin{proof}
The proof is provided in Appendix \ref{A1}.
\end{proof}

For large power constraint $\bar{P}_2$, \eqref{E68}-\eqref{E70} is
simplified as follows\vspace{-4pt}
\begin{align}\label{E200}
G_{M\!D,1}=(1-r_1)10\log_{10}{\bar{P}_2},\;\;\;\;G_{M\!D,2}=\frac{(br_1-1)}{br_1}10\log_{10}
{\bar{P}_2},\;\;\;\;G_{M\!D,3}=\frac{-1}{b}10\log_{10}{\bar{P}_2}.
\end{align}
As evident, $G_{\textit{MD}}$ depends on the power constraint and
bandwidth expansion ratio. One sees that both $G_{M\!D,1}$ and,
$G_{M\!D,2}$ for $br_1>1$ are increasing functions of the power. On
the other hand $G_{M\!D,3}$ is negative and is a decreasing function
of the power.

Tables \ref{T3} and \ref{T4} present the value of $G_{M\!D}$ for two
relatively large values of ${\bar{P}_2}$. One sees that, an adaptive
power scheme with constant channel rate performs better than an
adaptive channel rate scheme with constant power from mean
distortion perspective if $br_1>1$. As such, if we wish to control
only one parameter (power or rate) for efficient transmission in the
presence of source and channel variations, adapting power leads to a
superior mean distortion performance. Nonetheless, adapting rate
still provides performance gain.

The mean distortion exponent $\Delta_{\textit{MD}}$ of the presented
schemes which are derived in line with the proofs of the
Propositions \ref{P15}, \ref{P8}, \ref{P16} and \ref{P13} are
expressed in Table \ref{T5}. The mean distortion exponent indicates
the speed at which the mean distortion ($\rm dB$) reduces as the
average power (limit) ($\rm dB$) increases. Therefore, as evident,
this speed improves as bandwidth expansion ratio $b$ increases with
SCORPA, COPACR and CRCP schemes, while it is fixed to 1 with
SCORACP. Furthermore, the $\Delta_{\textit{MD}}$ obtained by all
proposed schemes is independent of the average power limit
$\bar{P}$. It is noteworthy that for $b=1$, $\Delta_{\textit{MD}}$
of COPACR is lower than that of SCORACP, while this is reverse for
$br_1>1$, where $r_1$ denotes the COPACR buffer unconstrained
multiplexing gain. Thus, it is evident that for large power
constraint and $br_1>1$ with buffer unconstrained scenario, the
adaptive power scheme with optimized fixed channel rate outperforms
the adaptive rate with fixed power design scheme from mean
distortion exponent aspect. The results in Tables \ref{T3} to
\ref{T5} indicate that from the perspective of mean distortion, for
delay-limited communication of quasi-stationary sources, CRCP scheme
may not be an appropriate design.
\begin{table}[!t]
\caption{{\small{Asymptotic mean power gain $G_{M\!D}$ of scheme 1
with respect to scheme 2 for source U defined in Section \ref{SVI},
unlimited buffer size, $K=2$ and $\bar{P}_2=40\rm dB$, where
$\bar{P}_2$ is the power limit for the scheme 2. } }} \label{T3}
\centering
\begin{tabular}{|c|c|c|c|c|c|c|}
\hline $G_{M\!D}$&Scheme 1&Scheme 2& $b=1$& $b=2 $& $b=4$& $b=6$\\
\hline $G_{M\!D,1}$&SCORPA&COPACR&$5.27$&$5.74$&$6.02$&$6.10$\\
\hline $G_{M\!D,2}$&
COPACR&SCORACP&$-0.008$&$9.78$&$17.46 $&$19.93$\\
\hline $G_{M\!D,3}$&
CRCP&SCORACP&$ -28.34$&$ -24.97$&$-19.45$&$-17.67$\\
\hline
\end{tabular}
\end{table}
\begin{table}[!t]
\caption{\small{Asymptotic mean power gain $G_{M\!D}$ of scheme 1
with respect to scheme 2 for source U defined in Section \ref{SVI},
unlimited buffer size, $K=2$ and $\bar{P}_2=45 \rm dB$, where
$\bar{P}_2$ is the power limit for the scheme 2.} } \label{T4}
\centering
\begin{tabular}{|c|c|c|c|c|c|c|}
\hline $G_{M\!D}$&Scheme 1&Scheme 2& $b=1$& $b=2 $& $b=4$& $b=6$\\
\hline $G_{M\!D,1}$&SCORPA&COPACR&$5.84$&$6.39$&$6.76$&$6.87$\\
\hline $G_{M\!D,2}$&
COPACR&SCORACP&$-0.05$&$ 11.91$&$21.00 $&$ 23.95$\\
\hline $G_{M\!D,3}$&
CRCP&SCORACP&$-32.27$&$-27.47$&$-20.71$&$-18.50$\\
\hline
\end{tabular}
\end{table}
\begin{table*}[t]
\caption{\small{Mean distortion exponent $\Delta_{\textit{MD}}$ of
the proposed schemes for source U.}}\label{T5} \centering
\begin{tabular}{|c||c|c|c|c|}
\hline Scheme   &SCORPA& COPACR& SCORACP& CRCP\\
\hline $\Delta_{\textit{MD}}$&$b$   &$br_1$  &1       &$\frac{b}{b+1}$\\
\hline
\end{tabular}
\end{table*}
\section{Performance evaluation}\label{SVI}
In this section, the performance of the proposed schemes are
compared through numerical computations and evaluated for different
source variances and frame sizes. To this end, we consider Rayleigh
fading channel and two quasi-stationary sources with $N_s=9$ where
the variance of the source in the state $s$ is
$\sigma^2_s(s)=(1+(s-1)^2):\forall{s}\in\{1,...,N_S\}$. For the
first source, labeled as U, the probability of being in different
states is considered uniform which results in
$\text{E}_s[\sigma_s^2]=23.66$. For the second source, G, the said
distribution follows a discrete Gaussian distribution with a mean
$5.49$ and a variance $2.52$ so as to have
$\text{E}_s[\sigma_s^2]=23.66$. We also consider a stationary source
D with $\sigma^2_s=23.66$, which is obviously equal to
$\text{E}_s[\sigma_s^2]$, same as expected variances of U and G. For
fair comparisons, we assume that $\tilde{B}_{\rm{max}}$ increases
linearly with $K$ such that $B_{\rm{max}}$ is independent of $K$.
%
%


Fig. \ref{F3} depicts reconstructed SNR (RSNR) performance of the
presented schemes defined as
$10\log_{10}\frac{\text{E}_s[\sigma^2_s]}{\text{E}[D]}$ with respect
to the power constraint $\bar{P}$ for different values of $K$. As
observed, RSNR improves as the source changes faster. In fact, the
delay due to buffering of the source blocks in a frame allows us to
use source diversity. However, the speed of this improvement
decreases as $K$ increases. Our simulations for the quasi-stationary
source U indicate that buffering of more than $K=4$ blocks does not
provide additional performance improvements.
\begin{figure}[t!]
     \centering
     \subfloat[]{\includegraphics[width=3.1in]{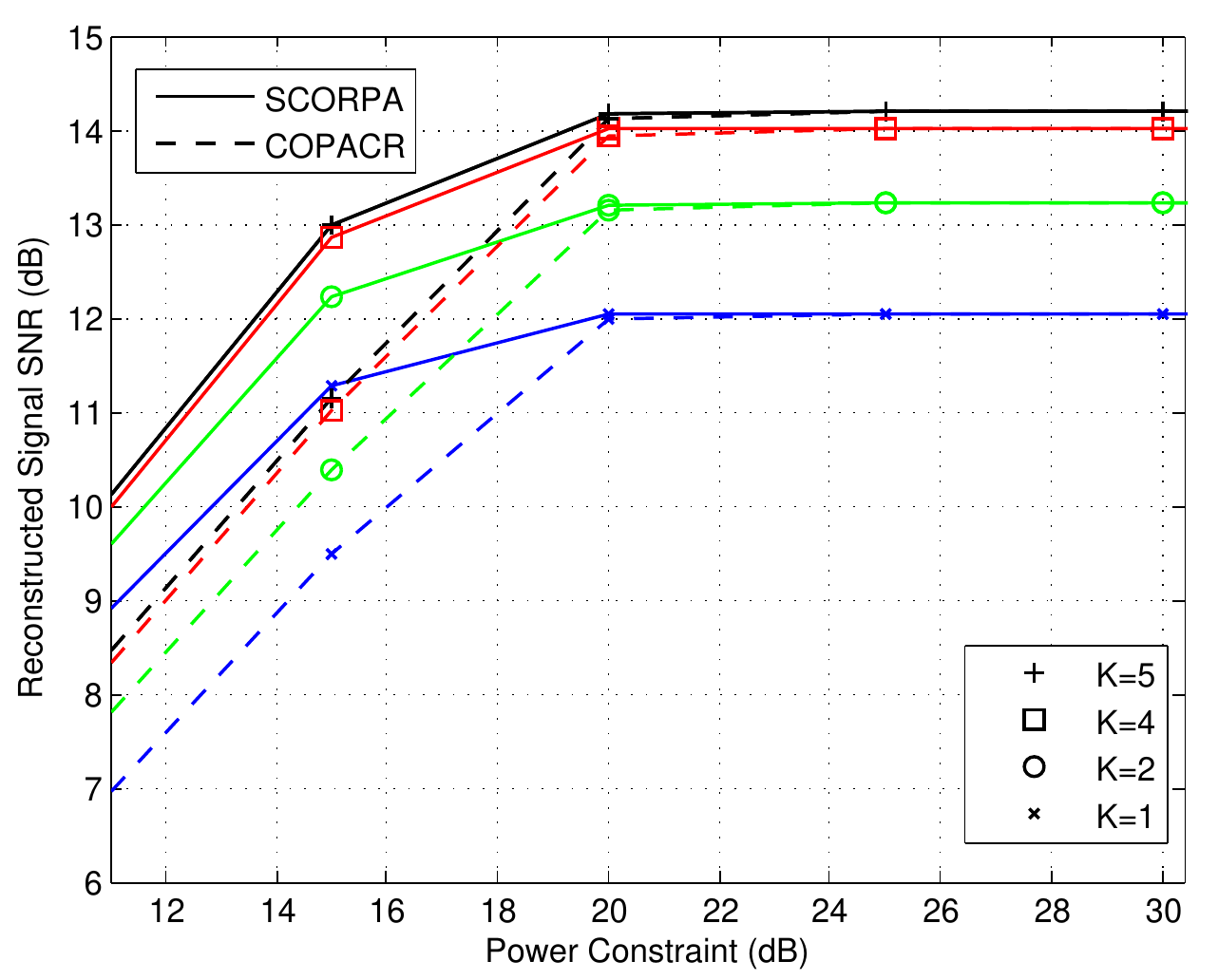}}
     \quad
     \subfloat[]{\includegraphics[width=3.1in]{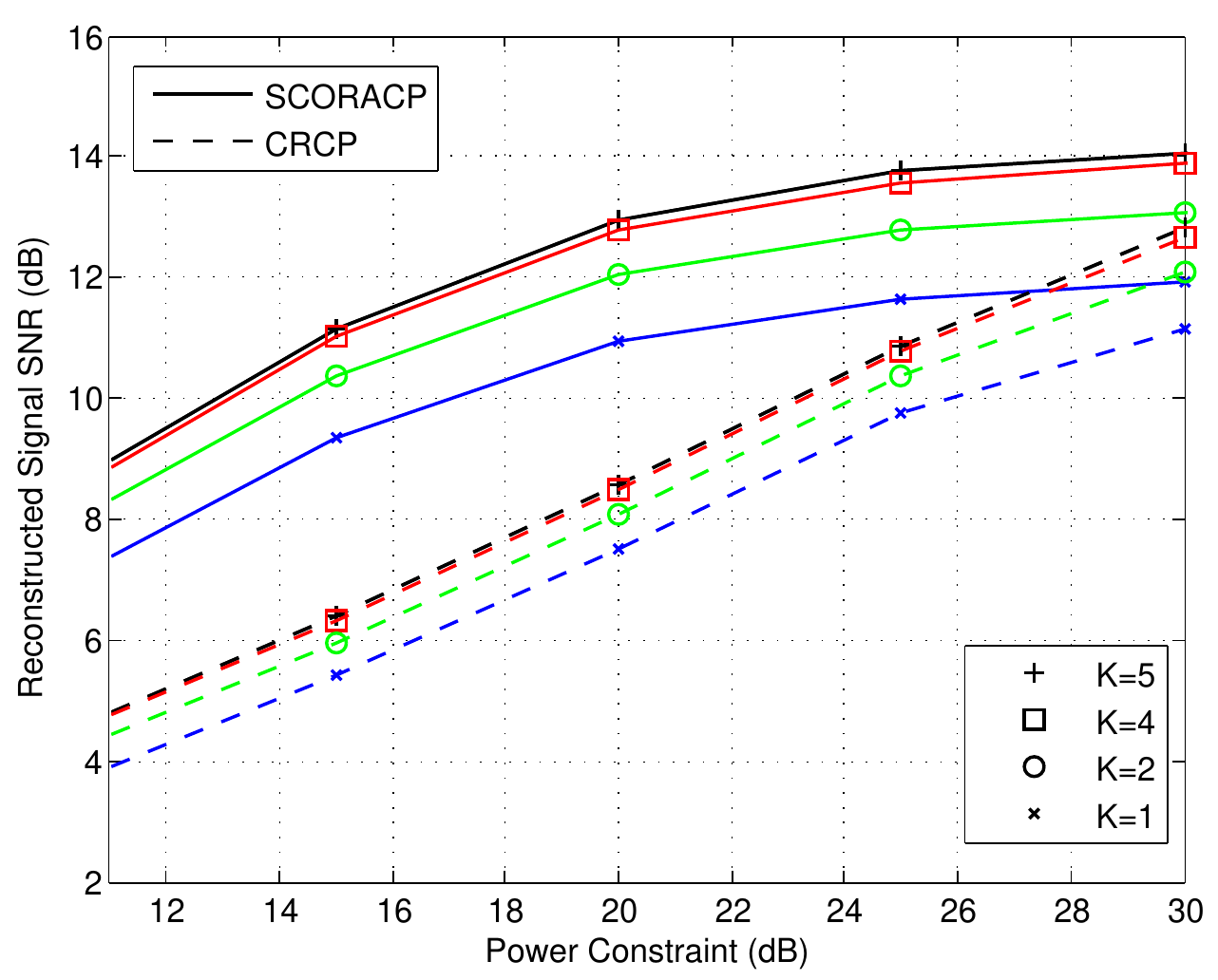}}
     \caption{\small{RSNR of (a) SCORPA and COPACR schemes (b) SCORACP and CRCP schemes, versus power constraint for different $K$, $B_{\rm{max}}=4$ and $b=1$.}}\label{F3}
     \label{steady_state}
\end{figure}

Figs. \ref{F5} and \ref{F8} demonstrate the RSNR performance of the
proposed schemes for bandwidth expansion ratios $b=1$ and $b=2$. As
noted in Section \ref{SVIB}, for $b=1$ and large $B_{\rm{max}}$,
here $B_{\rm{max}}=20$, SCORACP outperforms COPACR. This is while
for $b>1$ or limited $B_{\rm{max}}$, COPACR outperforms SCORACP for
large enough power constraint. A larger bandwidth expansion ratio,
$b$, corresponds to larger number of channel uses per source sample
and hence, as the results confirm, leads to improved RSNR
performance.
\begin{figure}[h!]
\centering
\subfloat[]{\includegraphics[width=3.1in]{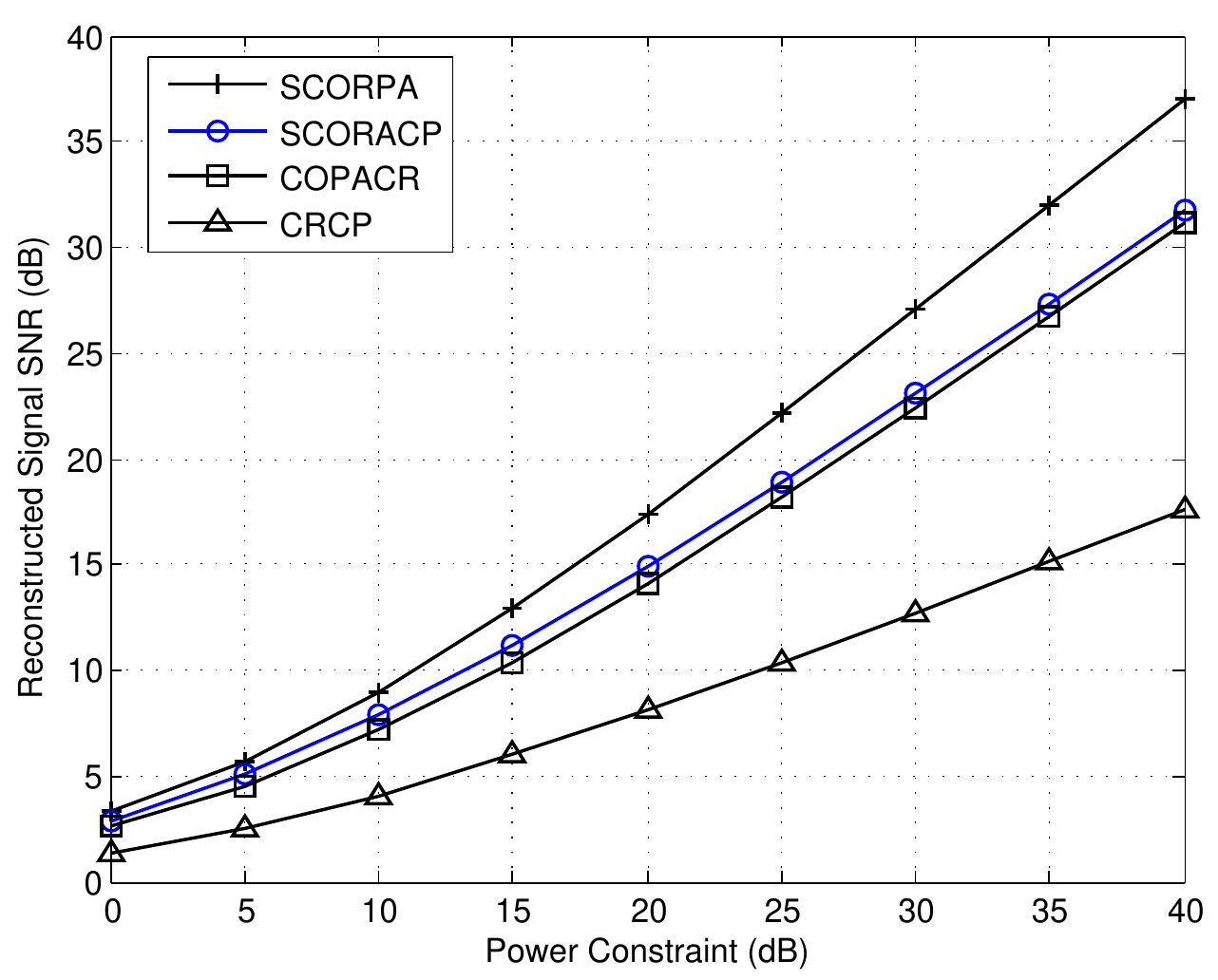}}\quad
\subfloat[]{\includegraphics[width=3.1in]{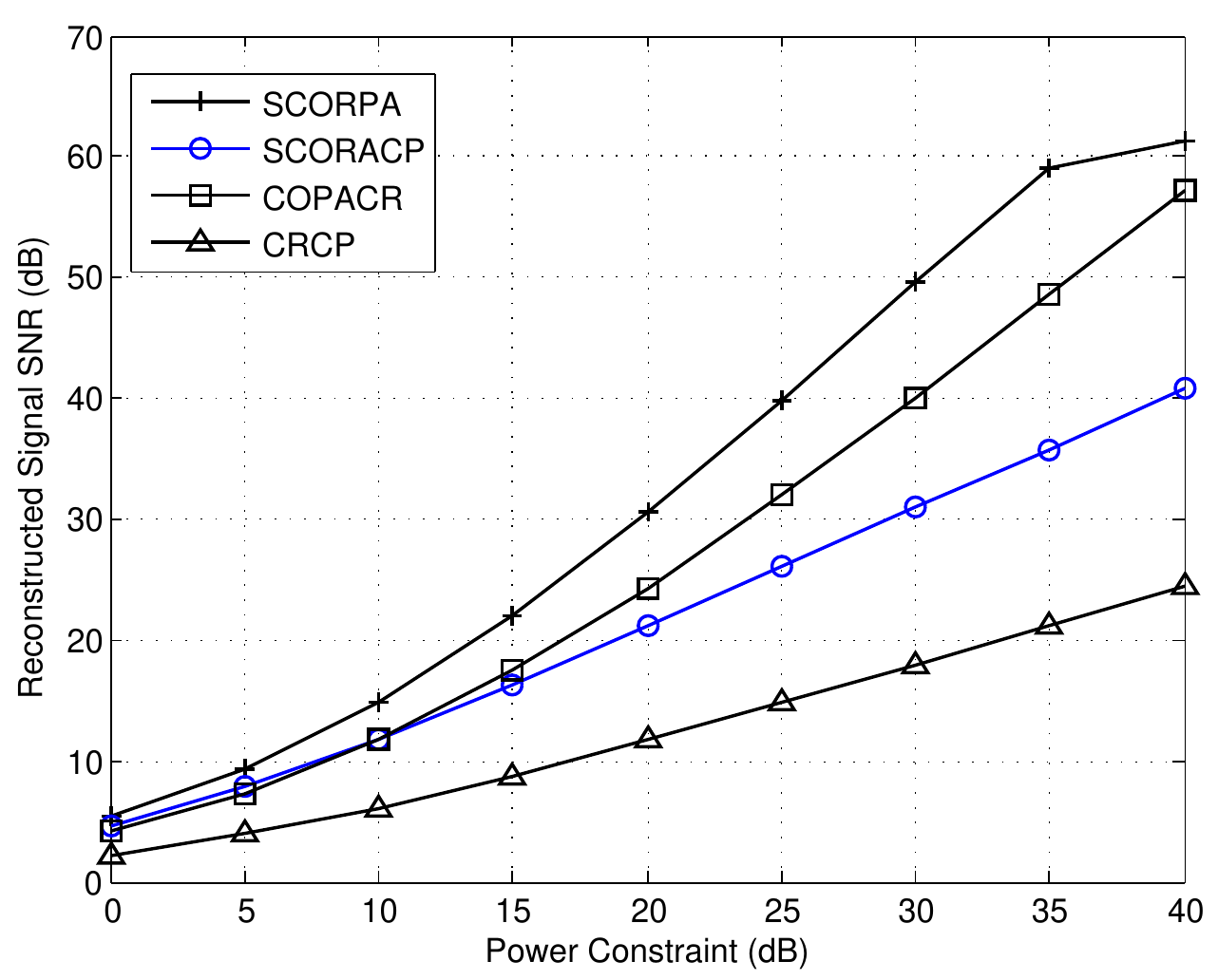}}
 \caption{\small{RSNR versus power constraint for $K=2$,
$B_{\rm{max}}=20$ and (a) $b=1$ (b) $b=2$.}} \label{F5}
\end{figure}

As observed in Fig. \ref{F5}, the proposed SCORPA scheme achieves an
asymptotic mean power gain of about $\rm5.81\, dB$ and $\rm5.91\,
dB$ with respect to COPACR, for $\bar{P}=\rm 40\, dB$ and $b=1$ and
$b=2$, respectively. In the same settings, the COPACR scheme
achieves asymptotic mean power gains of about $\rm\,-0.60 dB$ and
$9.39\rm\, dB$ with respect to SCORACP; and CRCP achieves gains of
$\rm\,-28.7 dB$ and $\rm -25.22\, dB$ with respect to SCORACP. Note
that $\bar{P}$ is the power limit for the second scheme in each
comparison. The results obtained from simulations and what is
reported in Table \ref{T3} from analyses match reasonably well given
the assumption of high average SNR considered in the analytical
performance evaluations.
\begin{figure}[h!]
\centering
 \subfloat[]{\includegraphics[width=3.1in]{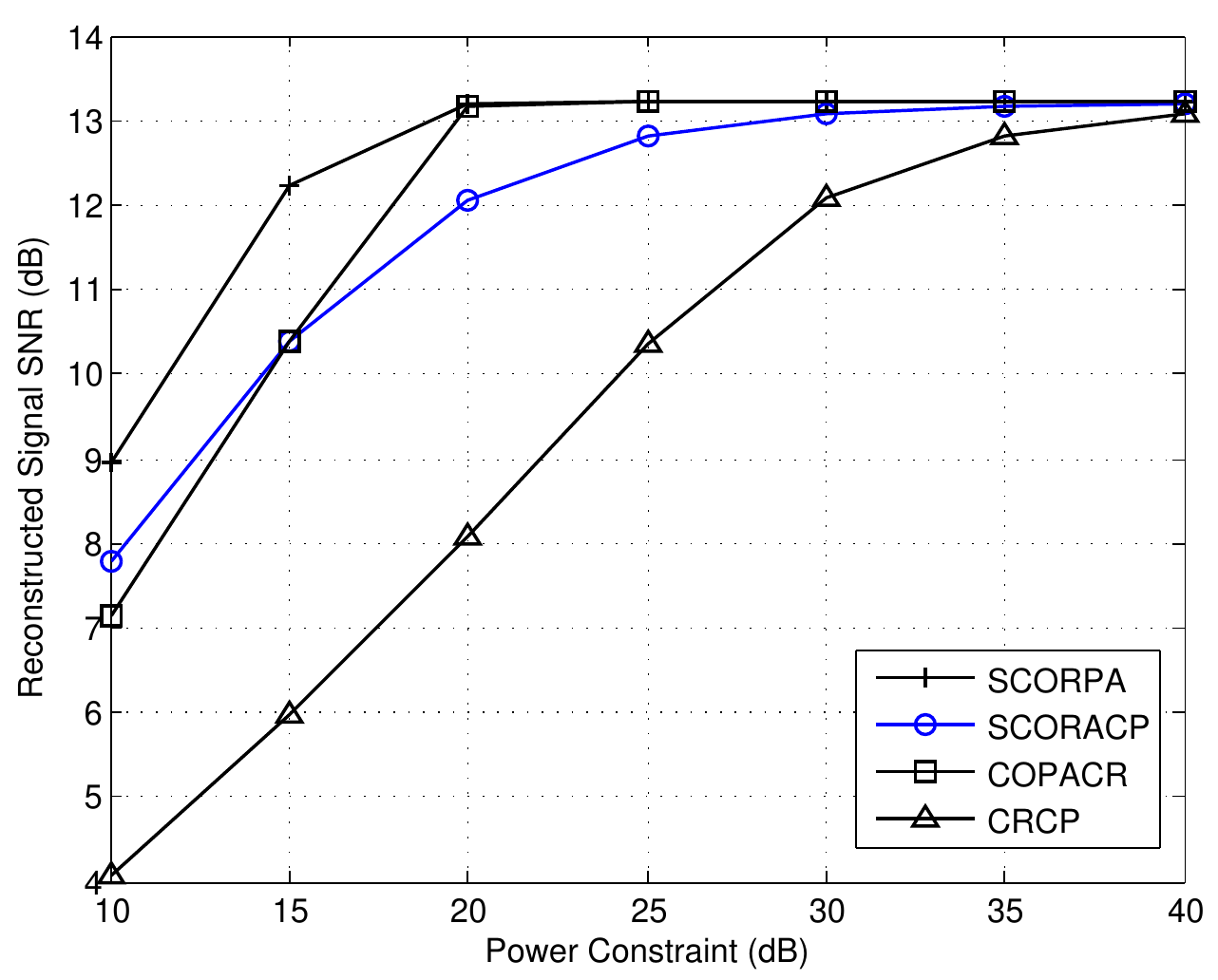}}\quad
\subfloat[]{\includegraphics[width=3.1in]{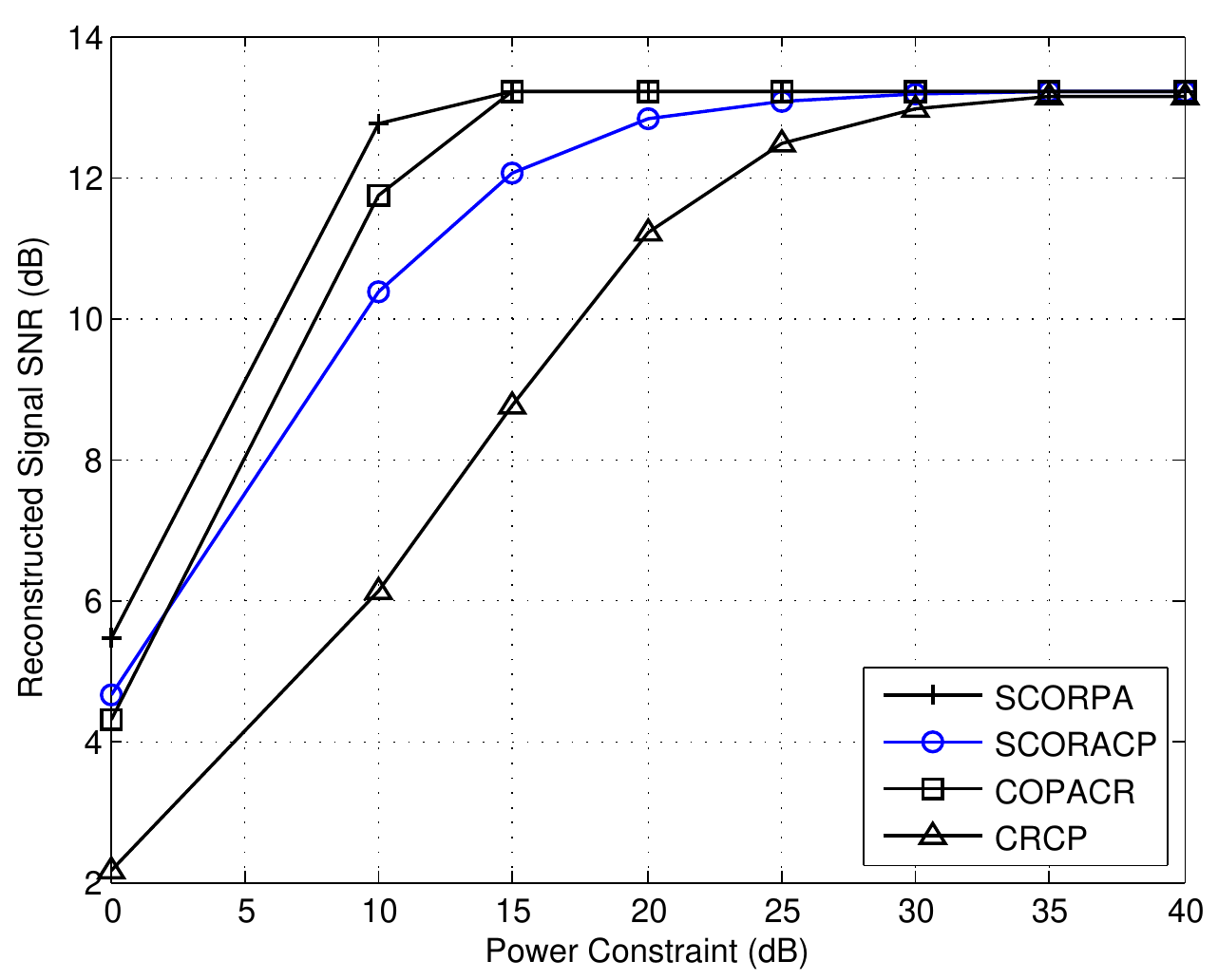}}
 \caption{\small{RSNR versus power constraint for $K=2$,
$B_{\rm{max}}=4$ and (a) $b=1$ (b) $b=2$.}} \label{F8}
\end{figure}

Figs. \ref{F5} and \ref{F8} also demonstrate the effect of
$B_{\rm{max}}$. As observed, a given frame (buffer) size,
$B_{\rm{max}}$, imposes a certain RSNR cap on the performance. As
power limit, $\bar{P}$, increases, the RSNR improves until it
saturates at this RSNR cap and any further increase of power will
not help with RSNR performance. This confirms the substantial impact
of buffer size on the performance. In the unsaturated regime, the
performance and the speed by which it improves with respect to the
power naturally depends on the system parameters and the rate and
power allocation strategy. As evident in \eqref{E4}, \eqref{E38},
\eqref{E40} and \eqref{E75}, the value of the RSNR cap depends on
the source statistics, $K$, and $B_{\rm{max}}$ and is independent of
the rate and power allocation strategy.

Fig. \ref{F9} demonstrates the RSNR performance of the presented
schemes for different sources. The results show that a larger source
diversity may be exploited as the non-stationary characteristics of
the source intensifies (from source D to U), and therefore RNR
increases in general. This is not only evident in the RSNR level at
which the performance statures (for large power constraints), it is
also visible in the unsaturated RSNR regime (low to medium
transmission power). Moreover, one sees that the relative
performance gains explained from using the proposed rate and/or
power adaptation strategies are attainable in the unsaturated regime
disregarding the source characteristics.
\begin{figure}[h!]
\centering
\subfloat[]{\includegraphics[width=3.1in]{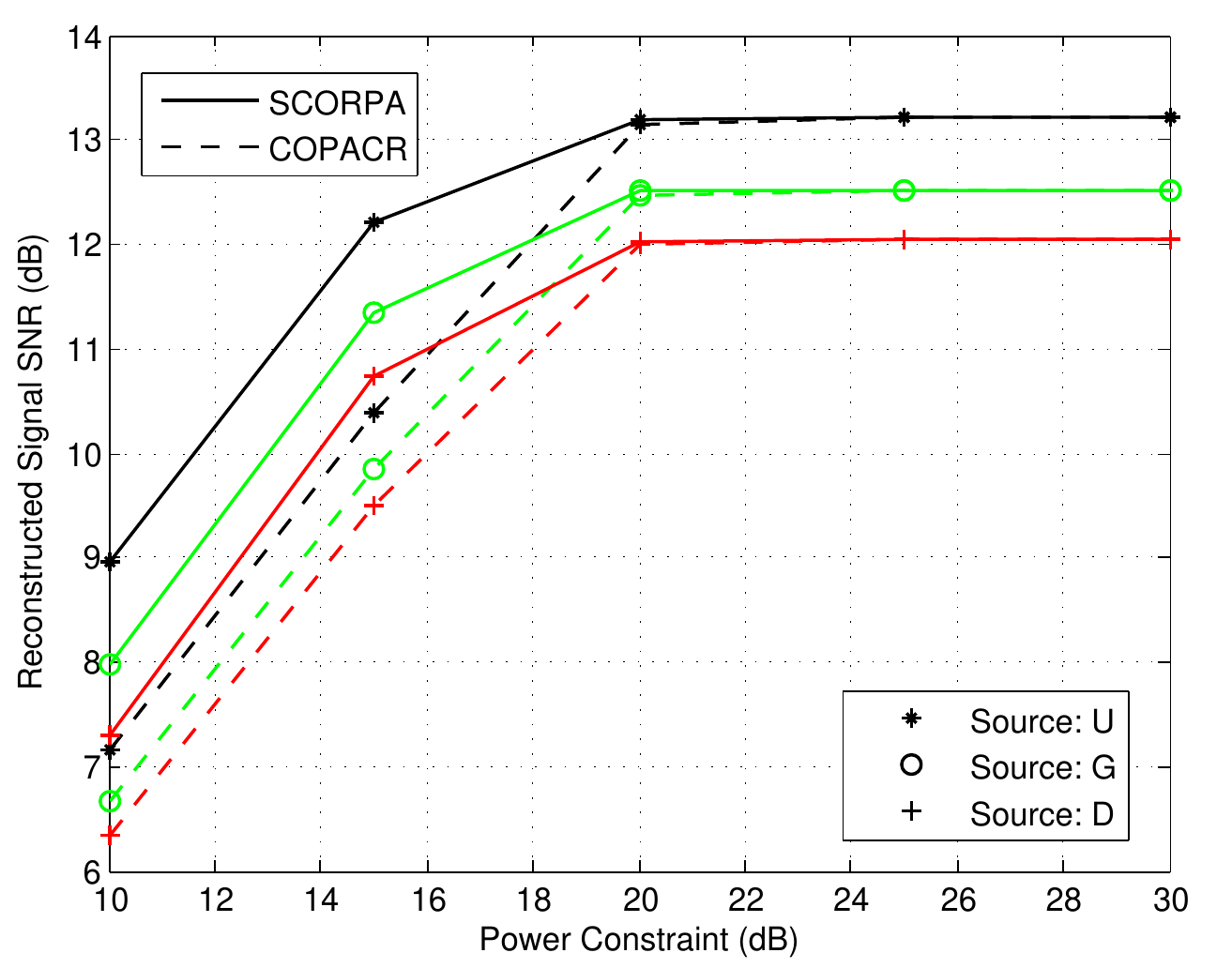}}\quad
\subfloat[]{\includegraphics[width=3.1in]{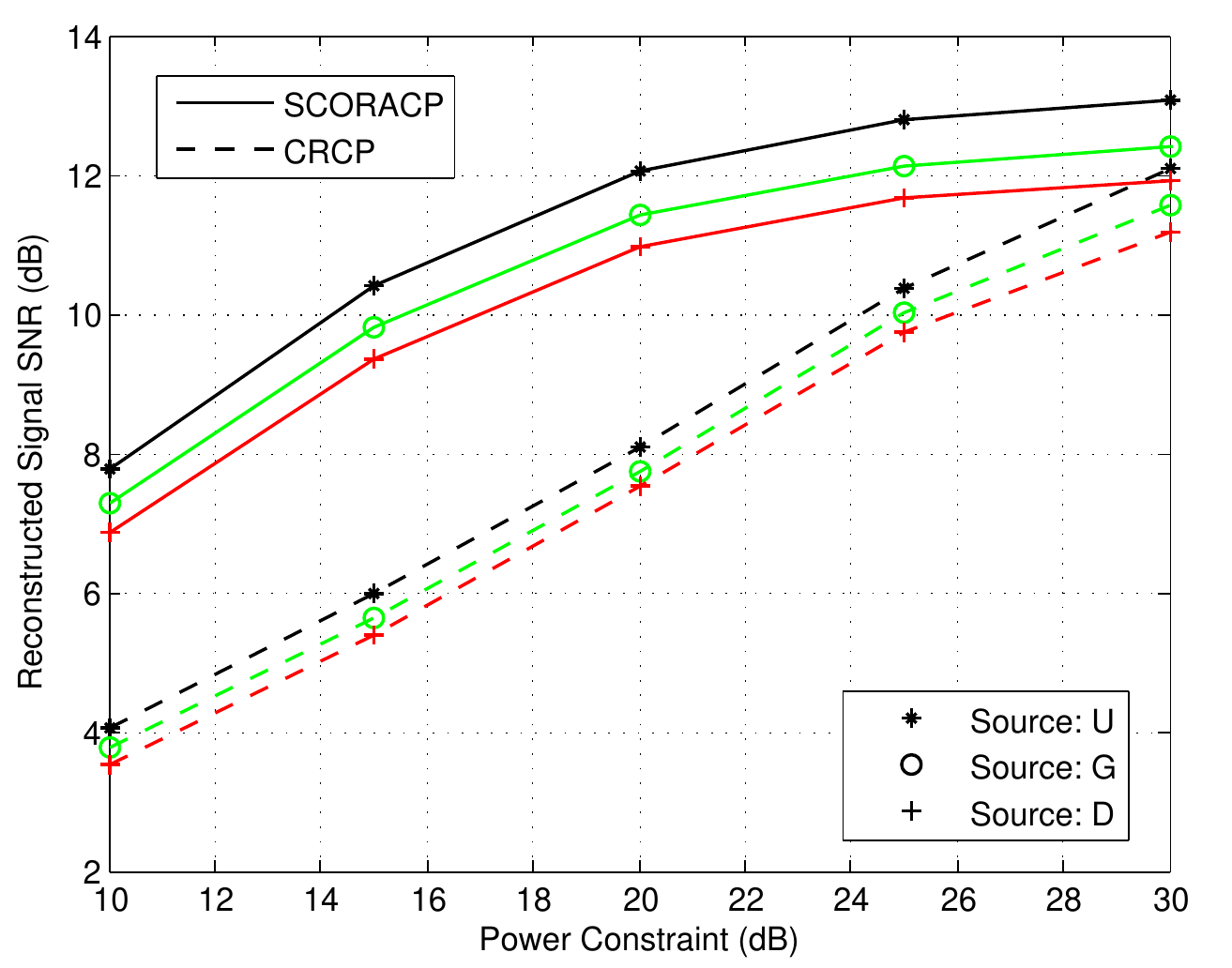}}
\caption{\small{RSNR of (a) SCORPA and COPACR schemes (b) SCORACP
and CRCP schemes, versus power constraint for three different
sources, $K=2$, $B_{\rm{max}}=4$ and $b=1$.}} \label{F9}
\end{figure}

\section{Conclusions}\label{SVII}
In this paper, we have considered delay-distortion-power trade-offs
in transmission of a quasi-stationary source over a block fading
channel when the buffer size is limited. Aiming at minimizing the
mean distortion, we have introduced two optimized power transmission
strategies as well as two other design schemes with constant
transmission power. In the high SNR regime, we have derived
different scaling laws involving mean distortion exponent and
asymptotic mean power gain. The analyses of the presented schemes
indicate that the buffer limit particularly affects the performance
as the average transmission power increases; and as such needs to be
carefully taken into account in the design. The proposed schemes
with buffering exploit the diversity gain due to non-stationary
characteristics of the source and channel variations to different
levels. Our studies confirm the benefit of power adaption along with
rate adaptation from a mean distortion perspective and for
delay-limited transmission of quasi-stationary sources with limited
buffer over wireless block fading channels.

Future research in this direction could investigate the potential
dependency of different source blocks in the design. Also, it is
interesting to model the characteristics of practical multimedia
coding standards within the proposed framework and hence quantify
the potential achievable performance gains. From a theoretical
perspective, one could also consider a multiuser setting and explore
design paradigms exploiting source and channel variations
(diversity) in a multiuser setting.

\appendices
\section{Proof of Proposition \ref{P1}}\label{A2}
Considering the power and buffer size constraints and equation
\eqref{E45}, it is necessary to set the power such that
$bC(\alpha,\gamma)\leq B_{\rm{max}}$. Thus, we may use the following
constraint\vspace{-6pt}
\begin{align}\label{E30}
bC(\alpha,\gamma)\leq B_{\rm{max}}
\end{align}
instead of the third constraint in Problem \ref{PR2}.

Due to the fact that in Problem \ref{PR2}, the objective and the
constraints are convex functions of $\gamma$ and $R_{s_j}$, we take
a Lagrange optimization approach. Hence, using
$C(\alpha,\gamma)=\log_2(1+\alpha\gamma)$, Problem \ref{PR2} may be
restated as follows\vspace{-6pt}
\begin{align}\label{E12}
\underset{\gamma,{R}_{s_j}}{\operatorname{min}}\;\biggl\{\text{E}_{\Sigma,\alpha}\left[\sum_{j=1}^{K}\sigma^2_{s_j}2^{-R_{s_j}}\right]+\lambda
E[\gamma]+\lambda_0
b\log_2(1+\alpha\gamma)+\lambda_1\left(\sum_{j=1}^{K}{R_{s_j}}-Kb\log_2(1+\alpha\gamma)\right)\biggr\}.
\end{align}
Differentiating \eqref{E12} with respect to $R_{s_j}$ and $\gamma$,
setting them to zero and noting the fact that $R_{s_j}$ and $\gamma$
are to be nonnegative, we obtain\vspace{-6pt}
\begin{align}\label{E3}
R_{s_j}=\left[\log_2{\frac{\sigma^2_{s_j}}{\lambda_2}}\right]^+\;\;\;\forall
j\in\{1,...,K\}
\end{align}
\begin{align}\label{E23}
\gamma=\left[\frac{\frac{\alpha\lambda_3bK}{\lambda}-1}{\alpha}\right]^+,
\end{align}\vspace{-4pt}
where $\lambda_2=\frac{\lambda_1}{\ln2}$ and
$\lambda_3=\lambda_2-\frac{\lambda_0}{K\ln2}$.

Here, we continue to solve the problem in different cases when the
constraint \eqref{E30} is active or inactive. Obviously, when the
constraint is inactive, i.e.,\vspace{-4pt}
\begin{align}\label{E15}
bC(\alpha,\gamma)< B_{\rm{max}}
\end{align}
and $\lambda_0=0$, we have $\lambda_3=\lambda_2$. Alternatively,
noting \eqref{E23} and \eqref{E15}, the constraint \eqref{E30} is
inactive if
\begin{align}\label{E16}
\frac{\alpha\lambda_2bK}{\lambda}<2^{\frac{B_{\rm{max}}}{b}}.
\end{align}
From
 \eqref{E23} it is seen that $\gamma=0$ if
$\frac{\alpha\lambda_2bK}{\lambda}\leq 1$ or equivalently,
$\lambda_2\leq\frac{\lambda}{\alpha bK}.$
Considering the forth constraint in Problem \ref{PR2}, $R_{s_j}=0$
for $\gamma=0$. Noting \eqref{E3}, $R_{s_j}=0,
\forall{j}\in\{1,...K\}$ if ${\lambda_2}\geq\sigma^2_{s_1}$.
Hence, the power and the rate are set to zero if $\frac{\lambda}{
bK\sigma^2_{s_1}}\geq\alpha$. Now assume
\begin{align}\label{E108}
\sigma^2_{s_{m+1}}<\lambda_2\leq\sigma^2_{s_m}.
\end{align}
As evident from \eqref{E3}, the rate is allocated to the $m$ blocks
out of $K$ blocks. From \eqref{NewCons}, \eqref{E3} and \eqref{E23},
we have
$\frac{\sigma^2_{s_1}...\sigma^2_{s_m}}{\lambda_2^m}=\left(\frac{\alpha\lambda_2bK}{\lambda}\right)^{bK}$
and therefore,
\begin{align}\label{E21}
\lambda_2=\sqrt[m+bK]{\left(\sigma^2_{s_1}...\sigma^2_{s_m}\right)}\sqrt[m+bK]{\left(\frac{\lambda}{\alpha
bK}\right)^{bK}}.
\end{align}
Noting \eqref{E108}, we obtain\vspace{-2pt}
\begin{align}\label{E207}
\sigma^2_{s_{m+1}}<\lambda_2=\sqrt[m+bK]{\sigma^2_{s_1}...\sigma^2_{s_m}}\sqrt[m+bK]{\left(\frac{\lambda}{\alpha
bK}\right)^{bK}}\leq\sigma^2_{s_m},
\end{align}
which indicates
$d_{1,m}\!=\!\frac{\!\lambda\sqrt[bK]{\sigma^2_{s_1}...\sigma^2_{s_m}}\!}{bK\!\sqrt[bK]{{\sigma_{s_m}}^{2(m+bK)}}}\!\leq\!\alpha<
d_{2,m}\!=\!\frac{\!\lambda\sqrt[bK]{\sigma^2_1...\sigma^2_{{s_m}}}\!}{bK\!\sqrt[bK]{{\sigma_{s_{m+1}}}^{2(m+bK)}}}.$
Obviously $d_{2,K}=\infty$. Given the value of
$\lambda_2$ computed in \eqref{E21}, \eqref{E16} may be rewritten as
follows
\begin{align}\label{E18}
\alpha<c_m=2^{\frac{m+bK}{m}\frac{B_{\rm{max}}}{b}}\frac{\lambda}{bK\sqrt[m]{\sigma^2_{s_1}...\sigma^2_{s_m}}}.
\end{align}

Now assume the constraint \eqref{E30} is active, i.e.,\vspace{-4pt}
\begin{align}\label{E20}
bC(\alpha,\gamma)= B_{\rm{max}},
\end{align}
$\lambda_0\neq 0$ and $\alpha\geq c_m$.
 Equivalently
 $\frac{\alpha\lambda_3bK}{\lambda}=2^{\frac{B_{\rm{max}}}{b}}.$
Therefore, utilizing \eqref{E3}, \eqref{E23} and \eqref{E20},
$\gamma$ and $R_{s_j}$ are given by\vspace{-4pt}
\begin{align}\label{E34}
\gamma=\left[\frac{2^{\frac{B_{\rm{max}}}{b}}-1}{\alpha}\right]^+
\;\;\;\;\text{if}\;\;\;\alpha\geq c_m
\end{align}
\vspace{-4pt}and
\begin{align}\label{E110}
R_{s_j}=\left[\log_2{\frac{\sigma^2_{s_j}}{{\lambda}_2}}\right]^+\;\;\;\forall
j\in\{1,...,K\}.
\end{align}
Consider $\sigma^2_{s_{n+1}}<{\lambda}_2\leq\sigma^2_{s_n}$.
As observed from \eqref{E110}, the rate is allocated to the $n$
blocks out of $K$ blocks. Therefore, from \eqref{NewCons},
\eqref{E34} and \eqref{E110}, we obtain
$\frac{\sigma^2_{s_1}...\sigma^2_{s_n}}{{\lambda}_2^n}=2^{KB_{\rm{max}}}.$
%
Hence, ${\lambda}_2$ is obtained as follows\vspace{-4pt}
\begin{align}\label{E109}
{\lambda}_2=\sqrt[n]{\frac{\sigma_{s_1}^2...\sigma^2_{s_n}}{2^{KB_{\rm{max}}}}}
\;\;\;\;\text{if}\;\;\;\alpha\geq c_m.
\end{align}
To distinguish between $\lambda_2$ in the two cases described, i.e.,
in \eqref{E21} and \eqref{E109}, we replace it with the new
multiplier $\mathbf{\tilde{\lambda}}_2$ in \eqref{E110} to
\eqref{E109}, when describing the Proposition \ref{P1}.

The parameter $\lambda$ is set to satisfy the power constraint
$\text{E}[\gamma]=\bar{P}$. Thus, \eqref{E48} is derived. Noting
exponentially distributed channel gain, we can obtain \eqref{E49}
and therefore the proof is complete.

\section{Proof of Proposition \ref{P15}}\label{A3}
We first consider COPACR with the buffer constrained scenario. Since
the buffer constrained mean distortion exponent,
$\Delta^{B_{\rm{max}}}_{\textit{MD}}$, of SCORPA is zero, we expect
that that $\Delta^{B_{\rm{max}}}_{\textit{MD}}$ of the other
proposed schemes are also equal to zero. However, we need to compute
the mean for large power constraint. Therefore, in the following we
write the mathematically steps to obtain
$\Delta^{B_{\rm{max}}}_{\textit{MD}}$. Due
$\Delta^{B_{\rm{max}}}_{\textit{MD}}$. Due to the fact that $R^*$ is
\eqref{E91} for large power $\bar{P}\rightarrow \infty$, we have
$\left(\frac{2^{R^*}-1}{q^*_1(R^*)}\right) \rightarrow 0$. Now using
\eqref{E98} in \eqref{E91}, we obtain\vspace{-6pt}
\begin{align}\label{E115}
\frac{2^{R^*}-1}{q^*_1(R^*)}=\exp\left(\frac{-\bar{P}}{2^{R^*}-1}\right).
\end{align}
From \eqref{E98} and \eqref{E115}, mean distortion in \eqref{E81} is
given by\vspace{-2pt}
\begin{align}\label{E116}
\text{E}[D] =\text{E}_s[\sigma_s^2]
\exp\left(\frac{-\bar{P}}{2^{R^*}-1}\right)+\text{E}_{\Sigma}\left[\frac{n\lambda(R^*)+\sum_{j=n+1}^{K}\sigma^2_{s_j}}{K}\right].
\end{align}
As evident, for large power the second term in \eqref{E116} is
dominant and therefore, in order to minimize $\text{E}[D]$ in
\eqref{E116}, it is necessary to set $R^*=\frac{B_{\rm{max}}}{b}$,
we have\vspace{-4pt}
\begin{equation}\label{E4}
\text{E}[D]=\text{E}_{\Sigma}\left[\frac{n\lambda(\frac{B_{\rm{max}}}{b})+\sum_{j=n+1}^{K}\sigma^2_{s_j}}{K}\right]=
\frac{1}{K}\text{E}_{\Sigma}\left[n\sqrt[n]{\frac{\sigma_{s_1}^2...\sigma^2_{s_n}}{2^{KB_{\rm{max}}}}}+\sum_{j=n+1}^{K}{\sigma^2_{s_j}}\right],
\end{equation}
where $n$ is an integer in $\{1,2,...,K\}$, such that
$\sigma^2_{s_{n+1}}<\sqrt[n]{\frac{\sigma_{s_1}^2...\sigma^2_{s_n}}{2^{KB_{\rm{max}}}}}\leq\sigma^2_{s_{n}}$
and hence
$\Delta^{B_{\rm{max}}}_{\textit{MD}}=\underset{\bar{P}\rightarrow\infty}{\lim}{-\frac{\ln{\text{E}[D]}}{\ln{\bar{P}}}}=0.$

Next, we consider the buffer unconstrained scenario. For large
$\bar{P}$ and $B_{\rm{max}}$, it is expected that $R$ is large
enough for limited source variances and from Proposition \ref{P7},
$\lambda(R)$ is very small. Thus, $n$ is set to $K$ and we have
\begin{align}\label{E204}
\lambda(R)=\sqrt[K]{\sigma^2_{s_1}\times...\times\sigma^2_{s_K}}2^{-bR}.
\end{align}
Replacing \eqref{E204} into \eqref{E81}, the mean distortion is
given by\vspace{-2pt}
\begin{align}\label{E205}
\text{E}[D]=\text{E}_s[\sigma_s^2]\left(1-\exp(-\frac{2^{R^*}-1}{q^*_1(R^*)})\right)+\exp\left(-\frac{2^{R^*}-1}{q^*_1(R^*)}\right)\text{E}_{\Sigma}\left[\sqrt[K]{\sigma^2_{s_1}\times...\times\sigma^2_{s_K}}\right]2^{-bR^*}.
\end{align}
In order for $\text{E}[D]$ to tend to zero, it is necessary that
$\frac{2^{R^*}-1}{q^*_1(R^*)}\rightarrow0$. Note that $2^{-bR^*}$
tends to zero for large $R^*$ and power. Now similar to the buffer
constrained scenario, \eqref{E115} is obtained. Thus, noting
\eqref{E98} and \eqref{E115}, the mean distortion in \eqref{E205} is
given by\vspace{-2pt}
\begin{align}\label{E67}
\text{E}[D] =\text{E}_s[\sigma_s^2]
\exp\left(\frac{-\bar{P}}{2^{R^*}-1}\right)+\text{E}_{\Sigma}\left[\sqrt[K]{\sigma^2_{s_1}\times...\times\sigma^2_{s_K}}\right]2^{-bR^*}.
\end{align}
Hence, $R^*$ is to be chosen such that $\text{E}[D]$ given in
\eqref{E67} is minimized. As stated, $R^*$ is a function of the
bandwidth expansion ratio $b$, the power constraint $\bar{P}$ and
the source variances in different blocks of a given frame. Fig.
\ref{F15} demonstrates $R^*$ in bits per channel use as a function
of $\bar{P}$ for different bandwidth expansion ratios $b$ and
exponentially distributed channel gain. It is evident that $R^*$
changes linearly with $\log_2\bar{P}$ for large $\bar{P}$, given $b$
and the source variances. Thus, $R^*$ may be described
by\vspace{-8pt}
\begin{equation}\label{E89}
R^*={r}_1\log_2\bar{P}+{r}_0,
\end{equation}
\begin{figure}[h!]
\centering
\includegraphics[width=3.1in]{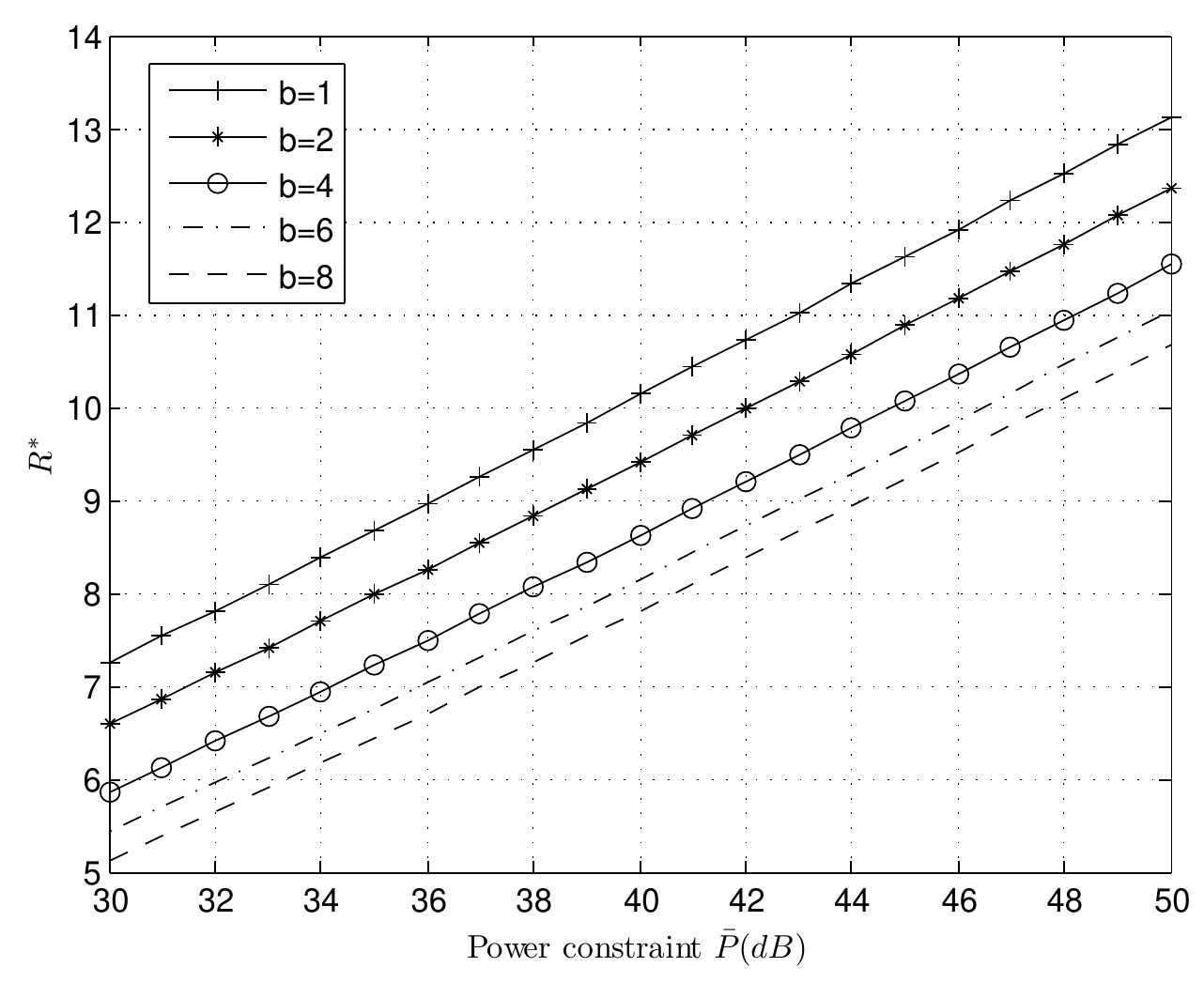} \caption{\small{$R^*$ versus $\bar{P}$(dB) in COPACR scheme for different bandwidth expansion ratios $b$ over a block Rayleigh fading
channel and source U.}} \label{F15}
\end{figure}
where ${r}_1$ and ${r}_0$ are obtained by least square fitting as
shown in Table \ref{T1}(a). Note that the results in Fig. 1 are due
to the quasi-stationary Gaussian source U described in Section
\ref{SVI} and additive Gaussian channel noise $\mathcal{N}(0,1)$.
Nonetheless, based on our experiments changing the source or the
Gaussian noise statistics merely reflects in the fitting parameters
$r_0$ and $r_1$ and do not affect the linear shape of the curves.
The parameter $r_1$ is referred to as the buffer unconstrained
multiplexing gain. From \eqref{E67} and \eqref{E89}, we have
\begin{align}\label{E65}
\text{E}[D] =\text{E}_s[\sigma_s^2]
\exp\left(\frac{-\bar{P}}{2^{r_0}\bar{P}^{r_1}-1}\right)+
\text{E}_{\Sigma}\left[\sqrt[K]{\sigma^2_{s_1}\times...\times\sigma^2_{s_K}}\right]2^{-br_0}\bar{P}^{-br_1}.
\end{align}
As evident, to have $\text{E}[D]\rightarrow 0$ for large power, it
is necessary that the power exponent in the first and second terms
to be respectively less than 1 and more than zero. Thus,
$r_1\in[0,1)$, and therefore we can ignore the first term with
respect to the second to obtain
\begin{align}\label{E66}
\text{E}[D]
=\text{E}_{\Sigma}\left[\sqrt[K]{\sigma^2_1\times...\times\sigma^2_K}\right]\bar{P}^{-br_1}2^{-br_0}.
\end{align}
 Thus, $\Delta_{\textit{MD}}=\underset{\bar{P}\rightarrow\infty}{\lim}{-\frac{\ln{\text{E}[D]}}{\ln{\bar{P}}}}=br_1.$

\begin{table}
\caption{\small{Parameters (a) $r_0$ and $r_1$ in COPACR scheme (b)
$\tilde{r}_1$ and $\tilde{r}_0$ in CRCP scheme, for source U
described in Section \ref{SVI} with the buffer unconstrained
scenario and $K=2$ as a function of bandwidth expansion ratio $b$}.}
\centering \label{T1}
\subfloat[]{%
\begin{tabular}{|c||c|c|}
\hline $b$&$r_1$&$r_0$ \\ \hline
    $1$ &$0.90$ & $-1.83 $  \\\hline
    $2$ &$ 0.89$ & $-2.39 $  \\\hline
    $4$ &$ 0.88 $& $-3.03 $  \\\hline
    $6$ &$ 0.87$ & $-3.41 $  \\\hline
    $8$ &$ 0.86$ & $-3.68 $  \\\hline
\end{tabular}}
\qquad\qquad
\subfloat[]{%
\begin{tabular}{|c||c|c|}
\hline $b$&$\tilde{r}_1$&$\tilde{r}_0$\\
\hline 1 &  0.50 &  -0.19   \\\hline
    2    &  0.33    & 0.20\\\hline
    4    &  0.20    &0.32\\\hline
    6    & 0.14     &0.31  \\\hline
    8 &   0.11    &0.29\\\hline
\end{tabular}}
\end{table}
%

\section{Proof of Proposition \ref{P2}}\label{A4}
We solve the problem in two different cases 1) $B_{\rm{max}}\leq
bC(\alpha)$ and 2) $B_{\rm{max}}> bC(\alpha)$, where in case 1 and
2, respectively the first and the second constraints in \eqref{E120}
have to be satisfied. Cases 1 and 2 respectively are equivalent to
$\alpha\geq\frac{2^{B_{\rm{max}}}-1}{\bar{P}}$ and
$\alpha<\frac{2^{B_{\rm{max}}}-1}{\bar{P}}$.

Therefore, using Lagrange optimization approach, we have
\begin{align}\label{E41}
J={\text{E}_{\Sigma,\alpha}\left[\frac{\sum_{j=1}^{K}\sigma^2_s2^{-2R_{s_j}}}{K}\right]}+\lambda_1\sum_{j=1}^{K}{R_{s_j}}.
\end{align}
Differentiating $J$ with respect to $R_{s_j}$, setting it to zero
and noting the fact that $R_{s_j}$ is to be nonnegative, we obtain
\begin{align}\label{E121}
R_{s_j}=\left[\log_2\frac{\sigma_{s_j}^2}{\lambda}\right]^+,
\end{align}
where $\lambda=\frac{\lambda_1}{2\ln2}$. 
Satisfying the first constraint in \eqref{E120} in the case
$B_{\rm{max}}\leq C(\alpha)$ imposes
$\lambda=\sqrt[n]{\frac{\sigma_{s_1}^2...\sigma^2_{s_n}}{2^{KB_{\rm{max}}}}}$,
where $n$ is an integer in $\{1,2,...,K\}$ such that
$\sigma^2_{s_{n+1}}<\lambda\leq\sigma^2_{s_{n}}.$

Noting \eqref{E120} and \eqref{E121} in the case
$B_{\rm{max}}>C(\alpha)$, we obtain
$\lambda=\sqrt[m]{\frac{\sigma_{s_1}^2...\sigma^2_{s_m}}{(1+\alpha\bar{P})^{bK}}},$
where $m$ is an integer in $\{1,2,...,K\}$ such that
$\sigma^2_{s_{m+1}}<\lambda\leq\sigma^2_{s_{m}}.$
Thus, we have
$d_{1,m}=\frac{\frac{\sqrt[bK]{\sigma^2_{s_1}...\sigma^2_{s_m}}}{(\sigma^2_{s_m})^{\frac{m}{bK}}}-1}{\bar{P}}\leq\alpha<d_{2,m}=\frac{\frac{\sqrt[bK]{\sigma^2_{s_1}...\sigma^2_{s_m}}}{(\sigma^2_{s_{m+1}})^{\frac{m}{bK}}}-1}{\bar{P}},$
where $d_{2,K}=\infty$.

 To distinguish between $\lambda$ in the two
cases described, we replace it with the new multiplier
$\mathbf{\tilde{\lambda}}$ in case 1, when describing the
Proposition \ref{P2}.


\section{Proof of Proposition \ref{P13}}\label{A5}
 Using \text{E}[D] in \eqref{E123} for Rayleigh block fading
channel with the given optimized $R^*$, achieving \eqref{E74} is
straightforward. For $\bar{P}\rightarrow\infty$ with the buffer
constrained scenario, the solution to \eqref{E71}, $R^*$, have to be
set to $\frac{B_{\rm{max}}}{b}$. Therefore, noting \eqref{E71} and
\eqref{E74} we obtain
\begin{align}\label{E75}
\text{E}[D]&=\frac{1}{K}\text{E}_{\Sigma}\left[n\sqrt[n]{\frac{\sigma_{s_1}^2...\sigma^2_{s_n}}{2^{KB_{\rm{max}}}}}+\sum_{j=n+1}^{K}{\sigma^2_{s_j}}\right],
\end{align}
where $n$ is an integer in $\{1,2,...,K\}$ such that
$\sigma^2_{s_{n+1}}<\sqrt[n]{\frac{\sigma_{s_1}^2...\sigma^2_{s_n}}{2^{KB_{\rm{max}}}}}\leq\sigma^2_{s_{n}}.$
Therefore, $\Delta_{\textit{MD}}^{B_{\rm{max}}}=0.$

For $\bar{P}\rightarrow\infty$ and $B_{\rm{max}}\rightarrow\infty$
with bounded source variances, it is expected that $R^*$ is large.
Thus using reverse water-filling, the rate is allocated to all
blocks in a given frame, i.e., $n=K$, and we obtain
$\lambda(R^*)=\sqrt[K]{\sigma^2_{s_1}\times...\times\sigma^2_{s_K}}2^{-bR^*}$.
Replacing $\lambda(R^*)$ into \eqref{E74}, the mean distortion is
rewritten as follows
$\text{E}[D]=\text{E}_s[\sigma_s^2]\left(1-\exp(\frac{2^{R^*}-1}{\bar{P}})\right)+\exp(\frac{2^{R^*}-1}{\bar{P}})\text{E}_{\Sigma}\left[\sqrt[K]{\sigma^2_{s_1}\times...\times\sigma^2_{s_K}}\right]2^{-bR^*}.$
Note the fact that for large $R^*$, we have $2^{-2bR^*}=0$.
Therefore, in order for $\text{E}[D]$ to tend to zero for large
power, it is necessary to have
$\frac{2^{R^*}-1}{\bar{P}}\rightarrow0$. Thus, using \eqref{E98}, we
have
\begin{align}\label{E92}
\text{E}[D]&=\text{E}_s[\sigma_s^2]\frac{2^{R^*}-1}{\bar{P}}+\text{E}_{\Sigma}\left[(\sigma^2_{s_1}\times...\times\sigma^2_{s_K})^{\frac{1}{K}}\right]2^{-bR^*}.
\end{align}
We seek $R^*$ such that $\text{E[D]}$ is minimized. Fig. \ref{F16}
demonstrates $R^*$ in bits per channel use as a function of
$\bar{P}$ for different bandwidth expansion ratios $b$ and
exponentially distributed channel gain. The results are for source
$U$ described in Section \ref{SVI}, however, our experiments with
other sources reveal curves of similar behavior. It is evident that
$R^*$ changes linearly with $\log_2\bar{P}$ for large $\bar{P}$,
given $b$ and source variances. Thus, $R^*$ may be described
by\vspace{-4pt}
\begin{equation}\label{E77}
R^*=\tilde{r}_1\log_2\bar{P}+\tilde{r}_0.
\end{equation}
\begin{figure}[t!]
\centering
\includegraphics[width=3.1in]{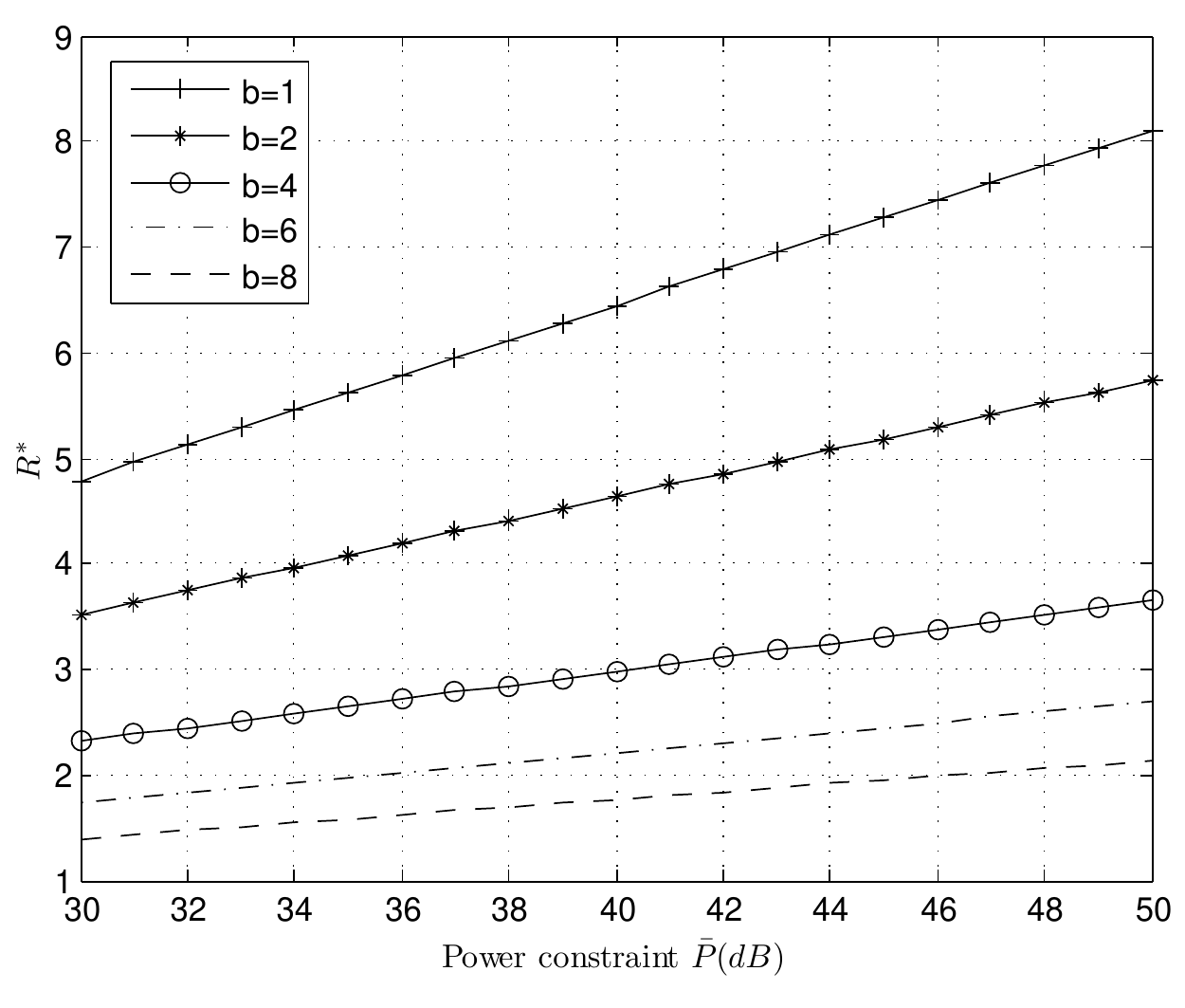} \caption{\small{$R^*$ versus $\bar{P}$(dB) in CRCP scheme for different bandwidth expansion ratios $b$ over a block Rayleigh fading
channel and source U.}} \label{F16}
\end{figure}
The parameter $\tilde{r}_1$ denotes the CRCP buffer unconstrained
multiplexing gain. From \eqref{E75} and \eqref{E77}, we obtain
$\text{E}[D] \!\!=\!\!\text{E}_s[\sigma_s^2]
{\bar{P}^{\tilde{r}_1-1}2^{\tilde{r}_0}}
\!\!+\!\!\text{E}_{\Sigma}\left[\sqrt[K]{\sigma^2_{s_1}\times...\times\sigma^2_{s_K}}\right]\bar{P}^{-b\tilde{r}_1}2^{-b\tilde{r}_0}.$
When $\bar{P}\rightarrow\infty$, for $\text{E}[D]\rightarrow 0$ we
need to have $\tilde{r}_1\in[0 \;1)$. Changing $\tilde{r}_1$ from 0
to 1, the power exponents ${\tilde{r}_1-1}$ and ${-b\tilde{r}_1}$
vary from $-1$ to $0$ and $0$ to $-b$, respectively. For large power
constraint and a given value of $\tilde{r}_1$, one of the two terms
$\bar{P}^{\tilde{r}_1-1}$ or $\bar{P}^{-b\tilde{r}_1}$ with the
larger exponent dominates. As the two exponents vary in opposite
directions when $\tilde{r}_1$ changes, the minimum value of the
dominating exponent occurs where the two exponents are equal. As a
result for maximum mean distortion exponent, we should
have\vspace{-10pt}
\begin{align} \label{E105}
\tilde{r}_1=\frac{1}{(b+1)}.
\end{align}
Thus, we obtain $\text{E}[D]
\!\!=\!\!\bar{P}^{-\frac{b}{b+1}}\biggl(2^{-b\tilde{r}_0}\text{E}_{\Sigma}\left[\!\!\sqrt[K]{\sigma^2_{s_1}\times...\times\sigma^2_{s_K}}\right]+2^{\tilde{r}_0}\text{E}_s[\sigma_s^2]\biggr)$
and
$\Delta_{\textit{MD}}\!\!=\!\!\underset{\bar{P}\rightarrow\infty}{\lim}\!{-\frac{\ln{\text{E}[D]}}{\ln{\bar{P}}}}\!\!=\!\!\frac{b}{b+1}.$
Table \ref{T1}(b) demonstrate $\tilde{r}_0$ and $\tilde{r}_1$
obtained by least squared fitting of the results in Fig. \ref{F16}
to the model in equation \eqref{E77}. As evident, the results
coincide with the analytical solution presented in \eqref{E105}.

\section{Proof of Proposition \ref{P9}}\label{A1}
The average power to asymptotically achieve a certain mean
distortion using SCORPA and COPACR schemes are denoted by
$\bar{P}_1$ and $\bar{P}_2$, respectively. Thus, we can use
\eqref{E84} to derive $G\!_{M\!D}$.

Noting \eqref{E55} and \eqref{E66}, we set
\begin{align}\label{E85}\vspace{-2pt}
\frac{1}{\bar{P}_1^b}\left(\text{E}_{\Sigma}\left[
\sqrt[K+bK]{\sigma^2_1...\sigma^2_K}\right]\right)^{b+1}\Gamma\left(\frac{1}{b+1},0\right)^{b+1}=\text{E}_{\Sigma}\left[\sqrt[K]{\sigma^2_1\times...\times\sigma^2_K}\right]\bar{P}_2^{-br_1}2^{-br_0}.
\end{align}
Therefore, we can derive \eqref{E68}. Achieving \eqref{E69} or
\eqref{E70} is straightforward, when we use \eqref{E66}, \eqref{E62}
and \eqref{E78} or \eqref{E62}, \eqref{E78} and \eqref{E65}; and
obtain the following\vspace{-6pt}
\begin{align}\label{E86}
\text{E}_{\Sigma}\left[\sqrt[K]{\sigma^2_1\times...\times\sigma^2_K}\right]\bar{P}_1^{-br_1}2^{-br_0}=\bar{P}_2^{-1}V_m,\;\;b>1
\end{align}
\vspace{-2pt}and\vspace{-6pt}
\begin{align}\label{E87}
\text{E}_{\Sigma}\left[\sqrt[K]{\sigma^2_1\times...\times\sigma^2_K}\right]\bar{P}_1^{-br_1}2^{-br_0}=\bar{P}_2^{-1}W_m,\;\;b=1\end{align}
\vspace{-6pt}or\vspace{-6pt}
\begin{align}\label{E90}
\bar{P}_1^{-\frac{b}{b+1}}\biggl(2^{-b\tilde{r}_0}\text{E}_{\Sigma}\left[\sqrt[K]{\sigma^2_1\times...\times\sigma^2_K}\right]+2^{\tilde{r}_0}\text{E}[\sigma^2]\biggr)=\bar{P}_2^{-1}V_m,\;\;b>1,
\end{align}
\vspace{-6pt}and\vspace{-6pt}
\begin{align}\label{E88}
\bar{P}_1^{-\frac{1}{2}}\biggl(2^{-\tilde{r}_0}\text{E}_{\Sigma}\left[\sqrt[K]{\sigma^2_1\times...\times\sigma^2_K}\right]+2^{\tilde{r}_0}\text{E}[\sigma^2]\biggr)=\bar{P}_2^{-1}W_m,\;\;b=1.
\end{align}

\ifCLASSOPTIONcaptionsoff
  \newpage
\fi

\bibliographystyle{Ieeetr}
\bibliography{IEEEabrv,Ref}

\end{document}